\documentclass{article}

\usepackage{arxiv}

\usepackage[utf8]{inputenc} 
\usepackage[T1]{fontenc}    
\usepackage{hyperref}       
\usepackage{url}            
\usepackage{booktabs}       
\usepackage{amsfonts}       
\usepackage{nicefrac}       
\usepackage{microtype}      
\usepackage{lipsum}		
\usepackage{graphicx}
\usepackage[numbers]{natbib}
\usepackage{doi}

\date{}

\renewcommand{\shorttitle}{ }

\usepackage{amsmath,amsfonts}
\usepackage{graphicx}

\usepackage{caption}
\usepackage{subcaption}

\usepackage{bm} 
\usepackage{bbm} 

\newcommand{\stab}{\operatorname{stab}} 
\newcommand{\ch}{\operatorname{CH}}     
\newcommand{\med}{\operatorname{med}}
\newcommand{\dist}{\operatorname{dist}}
\newcommand{\argmin}{\operatorname*{arg\,min}} 

\newcommand{\etal}{\textit{et al.}}

\usepackage{amsthm}
\theoremstyle{plain}
\newtheoremstyle{claimstyle}{\topsep}{\topsep}{}{0pt}{\sffamily}{. }{5pt plus 1pt minus 1pt}%
  {$\vartriangleright$ \thmname{#1}\thmnumber{ #2}\thmnote{ (#3)}}
\theoremstyle{plain}
\newtheorem{theorem}{Theorem}
\newtheorem{lemma}[theorem]{Lemma}
\newtheorem{corollary}[theorem]{Corollary}

\newtheorem{definition}[theorem]{Definition}
\newtheorem{conjecture}[theorem]{Conjecture}
\newtheorem{observation}[theorem]{Observation}
\theoremstyle{definition}

\theoremstyle{remark}

\newtheorem*{note*}{Note}

\newtheorem*{remark*}{Remark}

\newtheorem*{claim*}{Claim}

\title{Curve Stabbing Depth: Data Depth for Plane Curves%
   \thanks{A preliminary version of this work appeared in the proceedings of the 34th Canadian Conference on Computational Geometry (CCCG 2022) \cite{Durocher:2022}. This research was funded in part by the Natural Sciences and Engineering Research Council of Canada (NSERC).}}

\author{\hspace{1mm}Stephane Durocher \\
	Department of Computer Science\\
	University of Manitoba\\
	\texttt{stephane.durocher@umanitoba.ca} \\
	\And
	\hspace{1mm}Alexandre Leblanc \\
	Department of Statistics\\
	University of Manitoba\\
	\texttt{alex.leblanc@umanitoba.ca} \\
    \And
	\hspace{1mm}Spencer Szabados\thanks{Communicating author. Research initiated while at the University of Manitoba.} \\
	David R. Cheriton School of Computer Science\\
	University of Waterloo\\
	\texttt{sszabados@uwaterloo.ca} \\
}

\begin{document}
\maketitle

\begin{abstract}
Measures of data depth have been studied extensively for point data. Motivated by recent work on analysis, clustering, and identifying representative elements in sets of trajectories, we introduce {\em curve stabbing depth} to quantify how deeply a given planar curve $Q$ is located relative to a given set $\cal C$ of curves in $\mathbb{R}^2$. Curve stabbing depth evaluates the average number of elements of $\cal C$ stabbed by rays rooted along the length of $Q$. We describe an $O(n^3 + n^2 m\log^2m+nm^2\log^2 m)$-time algorithm for computing curve stabbing depth when $Q$ is an $m$-vertex polyline and $\cal C$ is a set of $O(n)$ polylines, each with $O(m)$ vertices. We analyze the newly proposed depth under common properties sought for point data, including equivariance under transformations, stability of the median element, and robustness to perturbations in data. Finally, we analyze a randomized algorithm for approximating the curve stabbing depth and contrast it to our deterministic algorithm.
\end{abstract}

\section{Introduction}
\label{sec:intro}
Processes that generate and require analyzing functional or curve data are becoming increasingly common within various domains, including medicine (e.g., ECG signals \cite{Ieva:2013} and analysis of nerve fibres in brain scans \cite{Micheaux:2021}), GIS techniques for generating and processing positional trajectory data (e.g., tracking migratory animal paths \cite{Buchin:2013}, air traffic control \cite{Chu:2021}, and clustering of motion capture data \cite{Durocher:2020}), and in the food industry (e.g., classification of nutritional information via spectrometric data \cite{Ferraty:2003}). In this paper, we consider depth measures for curve data.

Traditional depth measures are defined on multidimensional point data and seek to quantify the centrality or the outlyingness of a given object relative to a set of objects or to a sample population. Common depth measures include simplicial depth~\cite{Liu:1990}, Tukey (half-space) depth~\cite{Tukey:1975}, Oja depth~\cite{Oja:1983}, convex hull peeling depth~\cite{Barnett:1976}, and regression depth~\cite{Rousseeuw:1999}. See \cite{Liu:1999} and \cite{Serfling:2003} for further discussion on depth measures for multivariate point data. Previous work exists defining depth measures for sets of functions (so called functional data) \cite{Ferraty:2003, Cuevas:2007, Ieva:2013, Claeskens:2014}, often with a focus on classification. Despite the fact that curves can be expressed as functions (e.g., $[0,1]\to \mathbb{R}^d$), depth measures for functional data typically do not generalize to curves, as they are often sensitive to the specific parameterization chosen, being restricted to functions whose co-domain is $\mathbb{R}$, which can only represent $x$-monotone curves (e.g., $[0,1]\to \mathbb{R})$. 

Consequently, new methods are required for efficient analysis of trajectory and curve data. Recent work has examined identifying representative elements (e.g., finding a median trajectory \cite{Buchin:2013} or a central trajectory \cite{vanKreveld:2017} for a given set of trajectories) and clustering in a set of trajectories \cite{Buchin:2015,Durocher:2020}. 


%
%
%
In this work, we introduce {\em curve stabbing depth}, a new depth measure defined in terms of stabbing rays to quantify the degree to which a given curve is nested within a given set of curves. Our main contributions are: 
\begin{itemize}
\setlength\itemsep{0.1em} 
    \item In Section~\ref{sec:defs}, we define curve stabbing depth, a new depth measure for curves in $\mathbb{R}^2$, and describe a general approach for evaluating the curve stabbing depth of a given curve $Q$ relative to a set $\mathcal{C}$ of curves in $\mathbb{R}^2$.
    \item In Section~\ref{sec:alg}, we present an $O(n^3 + n^2m\log^2 m + nm^2\log^2 m)$-time algorithm for computing the curve stabbing depth of a given $m$-vertex polyline $Q$ relative to a set $\mathcal{P}$ of $n$ polylines in $\mathbb{R}^2$, each with $O(m)$ vertices.
    \item In Section~\ref{sec:properties}, we analyze properties of curve stabbing depth.
    \item In Section~\ref{sec:approx}, we examine randomized algorithms for approximating the curve stabbing depth of a given curve $Q$ relative to a set $\mathcal{C}$ of curves in $\mathbb{R}^2$.

\end{itemize}

\section{Preliminaries and Notation}
\label{sec:preliminaries}
In this section we introduce some preliminary definitions used throughout the remainder of the paper that the reader should familiarize themselves with along with the accompanying notation.

\subsection{Unparameterized Curve Space}
\label{sec:preliminaries.unparameterizedCurveSpace}

In fashion similar to de~Micheaux \etal.~\cite{Micheaux:2021}, we begin by defining a space of parameterization-free planar curves. Our goal is to reconcile the often informal interpretation of a geometric curve as a continuous set of points ``drawn'' in the plane, with more rigorous definitions curves that enable rigorous analysis, as done in \cite{Kurtek:2012}, \cite{Kemppainen:2017}.

For clarity and ease of delivery, let $F^{(1)}$ denote the space of piecewise differentiable functions with domain $[0,1]$ and codomain $\mathbb{R}^2$, each equipped with the standard Euclidean metric $\|\cdot\|_2$, which parameterize plane rectifiable curves. That is, for any $f\in F^{(1)}$, there exists a finite partition $0=t_0<t_1<\dots<t_n=1$ of $[0,1]$, such that $f'$ exists on each open interval $(t_i,t_{i+1})$ and is continuous on $[t_i,t_{i+1}]$, meaning that $\lim_{t\to t_i^+} f'(t)$ and $\lim_{t\to t_{i+1}^-} f'(t)$ both exist, and the arc length of $f$, denoted $L(f)$, is finite. See \cite{Marsden:1999}.

Moving forward, we use the term \emph{curve} to refer to the \emph{trace} (range) of a function $f:[0,1]\to \mathbb{R}^2$, which is defined as the locus of points in $\mathbb{R}^2$ given by
\begin{equation*}
    \Gamma_f = \{p\in\mathbb{R}^2\mid \exists t\in[0,1] \text{ s.t. } f(t) = p\}.
\end{equation*}

Working with traces is intuitive and suffices for the study carried out in this work; however, due to the parametric nature of these curves, each trace $\Gamma_f$ admits infinitely many functional representations in $F^{(1)}$, which need not be equivalent in general. E.g., consider two functions $f$ and $g$ with identical traces that contain a loop; $f$ may orbit this loop multiple times, switching directions between traversals before proceeding along the remainder of the curve, whereas $g$ might complete only a single unidirectional loop. 

As such, we consider the set of curves $T^{(1)}\subset F^{(1)}$, where for any $f\in T^{(1)}$ there exists a finite set $X\subset [0,1]$ of points, such that for any $\{x,y\} \subset [0,1] \setminus X$, 
$f(x) \neq f(y)$. 
The resulting functions in the set $T^{(1)}$ do not retrace arc segments along their trace, except for possibly finitely many points in $X$, which we call \emph{crossings}.

Define two functions, $f,g\in T^{(1)}$, to be {\em equivalent}, denoted $f\sim g$, if there exists a continuous monotonically non-decreasing homomorphism $\alpha:[0,1]\to[0,1]$ such that $f(t)=(g\circ \alpha)(t)$ for all $t\in[0,1]$. In this way, we can define the space of \emph{unparameterized plane curves} as the quotient space $\Gamma=T^{(1)}/\sim$, which consists of the set of  \emph{equivalence classes} of the form $[f] = \{g\in T^{(1)}\mid g\sim f\}$. The resulting space, $\Gamma$, achieves the desired effect of adequately distinguishing curves with identically shaped traces, while still enabling their evaluation via parameterizations. Specifically, each equivalence class of a non-zero arc length function $[f]$ gives rise to a \emph{representative element}, used to label the class, that is never locally constant; a curve defined by $f$ is \emph{locally constant} if for some non-empty interval $(a,b)\subseteq [0,1]$ and some $c\in \mathbb{R}^2$, $f(t) = c$ for all $t\in (a,b)$. This follows by a result of \cite{Burago:2001}. 
Moving forward, the notation for equivalence classes is omitted, and each class' representative element is taken without explicit mention. This space is then equipped with the Fr\'{e}chet distance, which for two curves $C_1$ and $C_2$ is defined as
\begin{equation*}
    \dist(C_1,C_2) = \inf_{\substack{c_1\in[C_1]\\ c_2\in[C_2]}}\sup_{t\in[0,1]}\|c_1(t)-c_2(t)\|.
\end{equation*}
The distance (or difference) between two points (degenerate curves of zero arc length) trivially reduce to the distance between the points. See \cite{Alt:1995} for further discussion of Fr\'{e}chet distance. The resulting metric space, $(\Gamma,\dist(\cdot,\cdot))$, is taken to be the underlying universe of plane curves we consider.

\subsection{Preliminary Definitions}
\label{sec:preliminaries.definitions}

Our algorithm presented in Section~\ref{sec:alg} applies firstly to  \emph{polylines} (or polygonal chains) in $\mathbb{R}^2$, and by extension, to traces whose components consist of polylines, a well-known class of curves within $\Gamma$ whose definition we now recall.

\begin{definition}[Polyline]
A {\em polyline} is a piecewise-linear curve consisting of the line segments ${\overline{p_1p_2}, \overline{p_2p_3}, \ldots , \overline{p_{m-1}p_m}}$ determined by the sequence of points $(p_1,p_2 \ldots, p_m)$ in $\mathbb{R}^2$.
\end{definition}

Additionally, our definition of curve stabbing depth refers to the notions of a stabbing ray and stabbing number, which are defined.

\begin{definition}[Stabbing Number]\label{def:stabbingNumber}
Given a ray $\overrightarrow{q_\theta}$ rooted at a point $q$ in $\mathbb{R}^2$ that forms an angle $\theta$ with the $x$-axis, the {\em stabbing number} of $\overrightarrow{q_\theta}$ relative to a set $\mathcal{C}$ of plane curves, denoted $\stab_{\cal C}(\overrightarrow{q_\theta})$, is the number of elements in $\mathcal{C}$ intersected by $\overrightarrow{q_\theta}$; i.e., 
\begin{equation}
    \stab_\mathcal{C}(\overrightarrow{q_\theta}) = |\{C\in\mathcal{C}\mid C\cap \overrightarrow{q_\theta}\neq \emptyset\}|.
\end{equation}
\end{definition}

General position for points and line segments is a fundamental concept throughout computational geometry, which is utilized for simplifying arguments by avoiding the discussion surrounding singular edge cases which are trivially dealt with in practice. In our study of smooth curves, we make use of a similar notion, which we call {\em subpath uniqueness} for curves, for specifying how ``nice'' curves intersect. 

\begin{definition}[Subpath Uniqueness]
Two curves $C_1,C_2\in\Gamma$ are said to be unique with respect to their subpaths, 
if for any two arcs with non-zero arc length, $C'_1 \subseteq C_1$ and $C'_2 \subseteq C_2$, $\dist(C'_1, C'_2) > 0$.
\end{definition}

That is, for any two curves $C_1$ and $C_2$ with unique subpaths,
$C_1 \cap C_2$ is a (possibly empty) discrete set of points in $\mathbb{R}^2$.

\section{Curve Stabbing Depth}
\label{sec:defs}

Having established required preliminary notions, we introduce the definition of our depth measure for plane curves, along with subsidiary definitions used during its analysis.

\subsection{Definition}
\label{sec:defs.definition}

\begin{definition}[Curve Stabbing Depth]\label{def:curveDepth}
Given a plane curve $Q \in \Gamma$ and a set $\mathcal{C} \subseteq \Gamma$ of plane curves, the {\em curve stabbing depth} of $Q$ relative to $\mathcal{C}$, denoted $D(Q,\mathcal{C})$, is 
\begin{equation}
\small
    D(Q,\mathcal{C}) = \frac{1}{\pi L(Q)}\int_{q\in Q}\int_{0}^{\pi} \min\{\stab_{\mathcal{C}}(\overrightarrow{q_\theta}),\stab_{\mathcal{C}}(\overrightarrow{q_{\theta+\pi}})\} \,d\theta\, dq,  \label{eq:depth}
\end{equation}
where $L(Q)=\int_{q\in Q} ds$ denotes the arc length of $Q$ as calculated by the line integral along the components of $Q$, and $q$ is a point along $Q$.
\end{definition}

We note that, as $L(Q)$ approaches zero (the curve $Q$ becomes a point $q$), the value
of Eq.~\eqref{eq:depth} approaches
\begin{equation}
\small
    D(q,\mathcal{C}) = \frac{1}{\pi}\int_{0}^{\pi}\min\{\stab_{\mathcal{C}}(\overrightarrow{q_\theta}),\stab_{\mathcal{C}}(\overrightarrow{q_{\theta+\pi}})\} \,d\theta. \label{eq:pointDepth}
\end{equation}

Curve stabbing depth Eq.~\eqref{eq:depth} corresponds to the average depth of all points $q \in Q$, where the depth of a single point $q$ relative to $\mathcal{C}$, given by Eq.~\eqref{eq:pointDepth}, is the average stabbing number in all directions $\theta$ around $q$, and for each $\theta$, either the stabbing number of the ray $\overrightarrow{q_\theta}$ or its reflection $\overrightarrow{q_{\theta+\pi}}$ is counted; i.e., 
\begin{equation}
	D(Q,\mathcal{C}) = \frac{1}{\pi}\int_{q\in Q}D(q,\mathcal{C}).
\end{equation}
This generalizes the one-dimensional notion of depth that counts the fewest number of elements in a set that are less than or greater than the query point (outward rank); taking $\mathcal{C}$ as a set of points in one dimension, $D(Q,\mathcal{C})$ coincides with common notions of depth, such as Tukey depth \cite{Tukey:1975}, Convex Hull Peeling Depth \cite{Barnett:1976}, Integrated Rank-Weighted Depth \cite{Durocher:2019}, and Majority Depth \cite{liu:1993}. Consequently, a set $\mathcal{C}$ can be ordered (its elements ranked) according to their individual depths, that is, according to the curve stabbing depth $D(C,\mathcal{C})$ of each $C\in \mathcal{C}$ relative to $\mathcal{C}$. 

\subsection{Wedges and Tangent Points}
\label{sec:defs.wedge}

We now introduce concepts used in the algorithm presented in Section~\ref{sec:alg} to calculate curve stabbing depth. 

As a ray $\overrightarrow{q_\theta}$ rotates about a point $q$, $\stab_{\cal C}(\overrightarrow{q_\theta})$ partitions the range $\theta \in [0,\pi)$ into intervals such that for all values of $\theta$ in a given interval, $\overrightarrow{q_\theta}$ intersects the same subset of $\cal C$. These intervals partition the plane around $q$ into {\em wedges}. We generalize this notion and define the wedges determined by a point $q$ relative to a set $\mathcal{C}$ of curves.


\begin{definition}[Wedge]\label{def:wedge}
The {\em wedge} of the curve $C$ relative to the point $q$, denoted $w(q,C)$, is the region determined by all rays rooted at $q$ that intersect $C$:
\begin{equation}
	w(q,C) = \bigcup_{\theta\in [0,2\pi)}\{\overrightarrow{q_\theta}\mid\overrightarrow{q_\theta} \cap C \neq \emptyset\}.
\end{equation}
\end{definition}

\begin{figure}[h!]
    \centering
    \includegraphics[scale=0.5]{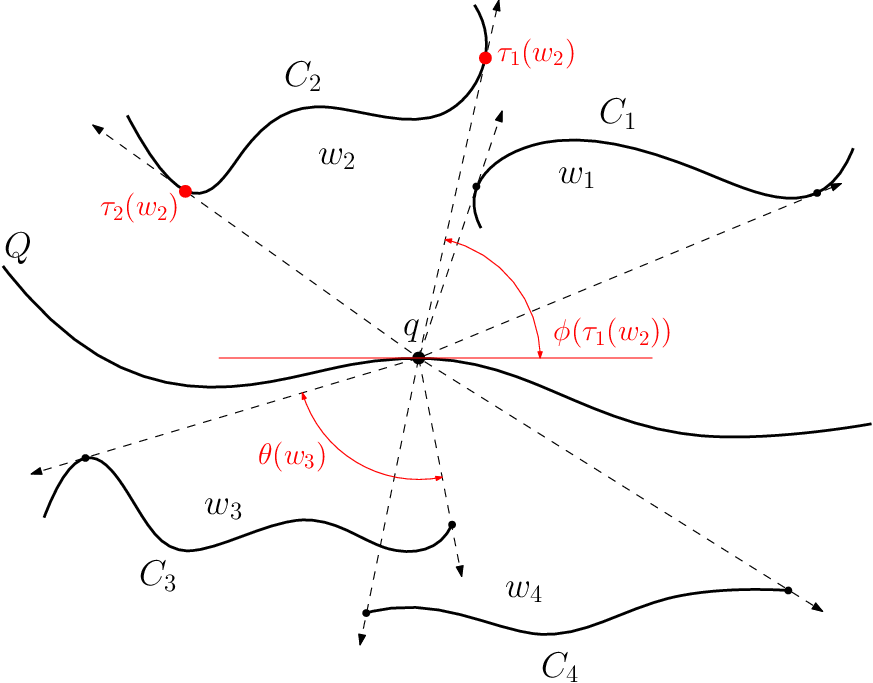}
     \caption{\small The set of wedges $w_1,w_2,w_3,$ and $w_4$ induced by curves $C_1,C_2,C_3$, and $C_4$ rooted at the point $q$ on the curve $Q$. Moving counterclockwise around $q$, the positive angle between $\tau_1(w_2)$ with the horizontal is indicated by $\phi(\tau_1(w_2))$, the tangent points of $w_2$ are labelled $\tau_1(w_2)$ and $\tau_2(w_2)$, and the internal angle of $w_3$ is highlighted by $\theta(w_3)$ with $\phi(w_3)=\phi(\tau_2(w_3))-\phi(\tau(w_3))$.}
    \label{fig:wedge}
\end{figure}

\begin{definition}[Tangent Points]\label{def:tangentPoints}
When $C \cup \{q\}$ is in general position in $\mathbb{R}^2$, the {\em tangent points} of the wedge $w=w(q,C)$, denoted $\tau(w)=\{\tau_1,\tau_2\}$, are those points of $C$ incident with the boundary of $w$; i.e., $\tau(w) = \partial w \cap C$, where $\partial w$ denotes the boundary of $w$. (If $C$ is a curve intersected by all rays rooted at $q$, the tangent points of $w(q,C)$ are taken to be coincident on $C$, with an internal wedge angle of $2\pi$ radians.) $\tau_1(w)$ denotes the tangent point that is the most clockwise of the two around $q$. The angles between the horizontal and tangent points of $w$ are denoted respectively by $\phi(\tau_1(w))$ and $\phi(\tau_2(w))$, with $\phi(w)$ denoting the interior angle of $w$. 
\end{definition}

See Figures~\ref{fig:wedge} and \ref{fig:wedgeAngle}.
The circular sequence of wedges determines an ordering of the curves stabbed about a given point $q$. Moreover, for a given set $\mathcal{C}$ of curves and associated set $\mathcal{W}_\mathcal{C}$ of wedges  rooted at a common point $q$, 
\begin{equation*}\label{eq:stabObs}
\small
\stab_{\mathcal{C}}(\overrightarrow{q_\theta})=|\{w\in \mathcal{W}_\mathcal{C} \mid \theta \in [\phi(\tau_1(w)),\phi(\tau_2(w))] \}|.
\end{equation*}

That is, $\stab_{\cal C}(\overrightarrow{q_\theta})$ is the number of wedges that contain the ray $\overrightarrow{q_\theta}$, where each wedge is associated with a curve in $\cal C$. See Figure~\ref{fig:wedge}.

\subsection{Computing Wedge Angles}\label{sec:wedgeAngles}
Our algorithm for computing curve stabbing depth requires calculating the interior angle $\phi(w)$ of each wedge $w$, which we now describe. We consider two cases for the relative positions of a given query line segment $Q$, a curve $C$, and the wedge $w(q,C)$ rooted at a point $q$: (Case 1) when $q\not\in \ch(C)$, where $\ch(C)$ denotes the convex hull of $C$, i.e., $Q$ does not pass through the interior of $C$, and (Case 2) when $q \in \ch(C)$. When points and curves are in general position with unique subpaths, $C$ cannot coincide with a bounding edge of $w$. See Figure~\ref{fig:wedgeAngle}.

\begin{figure}[h]
    \centering
    \begin{subfigure}{.4\textwidth}
        \centering
        \includegraphics[scale=0.6]{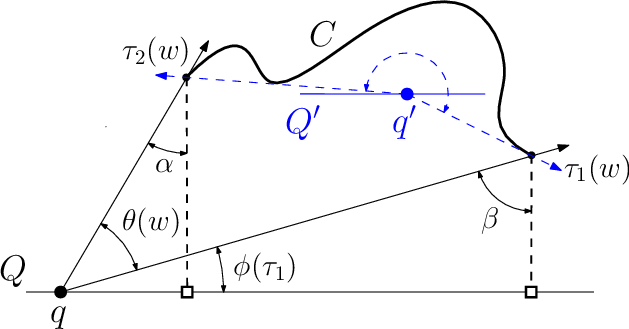}
        \caption{\small (Case 1) A curve $C$ positioned entirely above a query line segment $Q$, and (Case 2) positioned midway through a query curve $Q'$.}
        \label{fig:wedgeAngleAbove}
    \end{subfigure}
    \hfill
    \begin{subfigure}{.4\textwidth}
        \centering
        \includegraphics[scale=0.6]{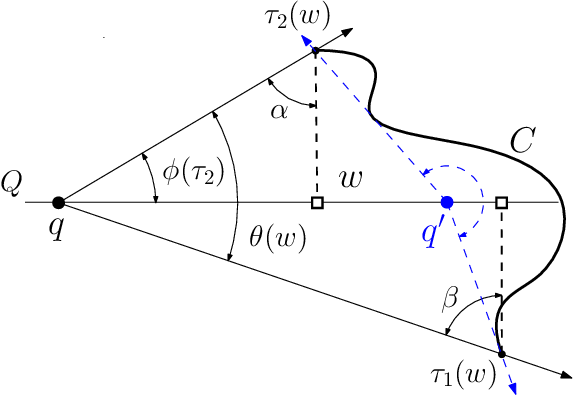}
        \caption{\small (Case 1) A curve $C$ crossing in front of a point $q$ along a query line segment $Q$ with $q\not\in\ch(C)$, and (Case 2) a point $q'\in\ch(C)$ further along $Q$.}
        \label{fig:wedgeAngleCross}
    \end{subfigure}
    \caption{\small Two ways a query line segment $Q$ and a wedge rooted at a point on $Q$ can be arranged under general position. Case~1 is drawn in black while Case~2 is outlined in blue.}
    \label{fig:wedgeAngle}
\end{figure}

In Case~1, when $C$ lies entirely above or below $Q$ the angles formed between the tangent points, root, and horizontal can be evaluated as
\begin{align}
\small
    \theta(w)+\phi(\tau_1) &= \frac{\pi}{2}-\alpha = \frac{\pi}{2}-\tan^{-1}\left(\left|\frac{q_x-\tau_2(w)_x}{\tau_2(w)_y-q_y}\right|\right),\\
    \phi(\tau_1) &= \frac{\pi}{2} - \beta = \frac{\pi}{2}-\tan^{-1}\left(\left|\frac{q_x-\tau_1(w)_x}{\tau_1(w)_y-q_y}\right|\right).
\end{align}
Where the interior angle of $w$ is found to be
\begin{equation}
    \small
    \theta(w) = \tan^{-1}\left(\frac{q_x-\tau_2(w)_x}{\tau_2(w)_y-q_y}\right)\\-\tan^{-1}\left(\frac{q_x-\tau_1(w)_x}{\tau_1(w)_y-q_y}\right).
    \label{eq:interiorAngleAbove}
\end{equation}
When $C$ crosses in front of $Q$, as illustrated in Figure~\ref{fig:wedgeAngleCross}, we calculate 
\begin{equation}
\small
    \theta(w) = \pi - \bigg|\tan^{-1}\left(\frac{q_x-\tau_2(w)_x}{\tau_2(w)_y-q_y}\right)\\
    + \tan^{-1}\left(\frac{q_x-\tau_1(w)_x}{\tau_1(w)_y-q_y}\right)\bigg|.
    \label{eq:interiorAngleCross}
\end{equation}
Once $q$ enters $\ch(C)$, we transition to Case~2, in which the calculations are analogous to those of Case~1, except for modifications needed to account for taking an angle greater than $\pi$ radians, as shown in Figure~\ref{fig:wedgeAngleCross} in blue. Every case considered by our algorithm reduces to Case~1 or Case~2. We sometimes limit discussion to instances of Case~1 depicted in Figure~\ref{fig:wedgeAngleAbove} to simplify the presentation; our results apply to both cases.

\begin{definition}[Circular Partition]\label{def:circularPartition}
The {\em circular partition} induced by the set of wedges $\mathcal{W}_\mathcal{C}=\{w_1,w_2,\dots,w_n\}$ rooted at a common point $q$ is the sequence $0 = \theta_0 < \theta_1 < \cdots < \theta_{4n} < 2\pi$ of angles, corresponding to the ordered sequence of bounding rays of wedges in $\mathcal{W}_\mathcal{C}$; i.e., it is the ordered sequence of values in
\begin{equation*}
    \bigcup_{w\in \mathcal{W_C}} \left\{\begin{array}{l}
        \phi(\tau_1(w)), \phi(\tau_1(w))+\pi \bmod 2\pi,\\ \phi(\tau_2(w)), \phi(\tau_2(w))+\pi\bmod 2\pi\end{array}\right\}.
\end{equation*}
Denote this sequence by $\sigma(\mathcal{W}_\mathcal{C})=(\theta_0,\theta_1,\dots,\theta_{4n})$. 
\end{definition}


\subsection{Wedge Partitioning and Invariance}
Applying Equation~\eqref{eq:stabObs} to Definition~\ref{def:circularPartition},  we arrive at the following observation (see Figure~\ref{fig:wedgePartition}):

\begin{figure}[!ht]
    \centering
    \includegraphics[scale=0.5]{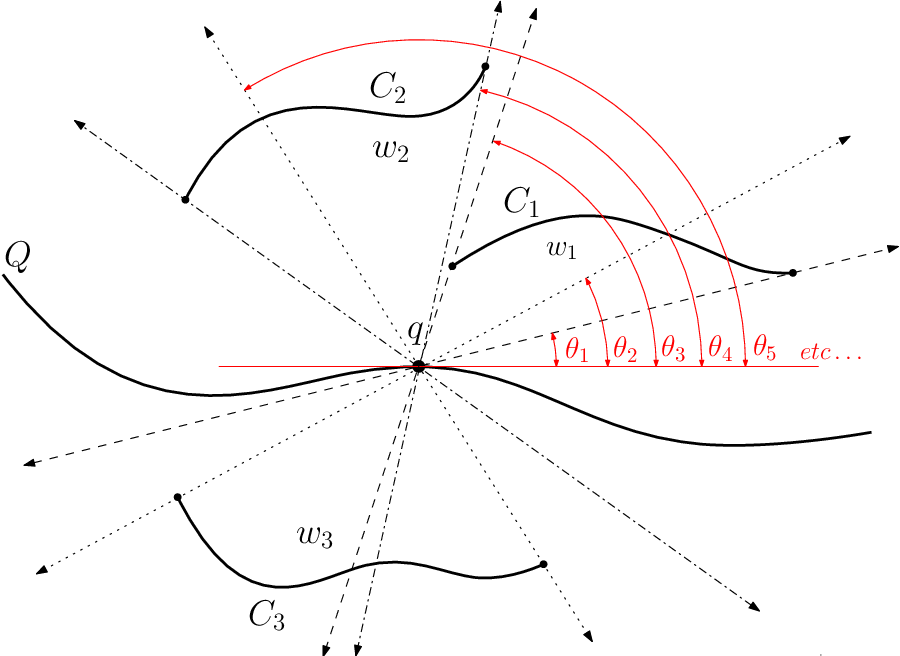}
    \caption{\small A configuration similar to that shown in Figure~\ref{fig:wedge} for three curves $C_1,C_2$, and $C_3$ is depicted, with their respective wedge boundaries extended through the origin. The circular partition induced is shown by the sequence of angles towards the right-hand side of the figure.}
    \label{fig:wedgePartition}
\end{figure}

\begin{observation}\label{obs:stabbing}
Given a set $\mathcal{W}_\mathcal{C}$ of wedges and induced partition $\sigma(\mathcal{W}_\mathcal{C})=(\theta_0,\theta_1,\dots,\theta_{4n})$ for a given point $q$ and set $\mathcal{C}$ of curves, for every $i \in \{0,1, \ldots, 4n-1\}$ and every $\phi_1,\phi_2 \in (\theta_i,\theta_{i+1})$, the set of curves in $\mathcal{C}$ intersected by $\overrightarrow{q_{\phi_1}}$ is the same as that intersected by $\overrightarrow{q_{\phi_2}}$.
\end{observation}

Observation~\ref{obs:stabbing} remains true for any point $q$ in general position relative to $\cal C$, where $q$ moves along $Q$ within a bounded neighbourhood: given a curve $Q$ and a set $\cal C$ of curves, for each point $q$ in general position on $Q$, the relative ordering of wedge boundaries in the circular partition of $q$ remains unchanged when $q$ moves along some interval of $Q$. By partitioning $Q$ into such cyclically invariant segments, this property allows us to calculate the curve stabbing depth of $Q$ relative to $\cal C$ discretely. More formally:

\begin{figure}[!ht]
    \centering
    \includegraphics[scale=0.5]{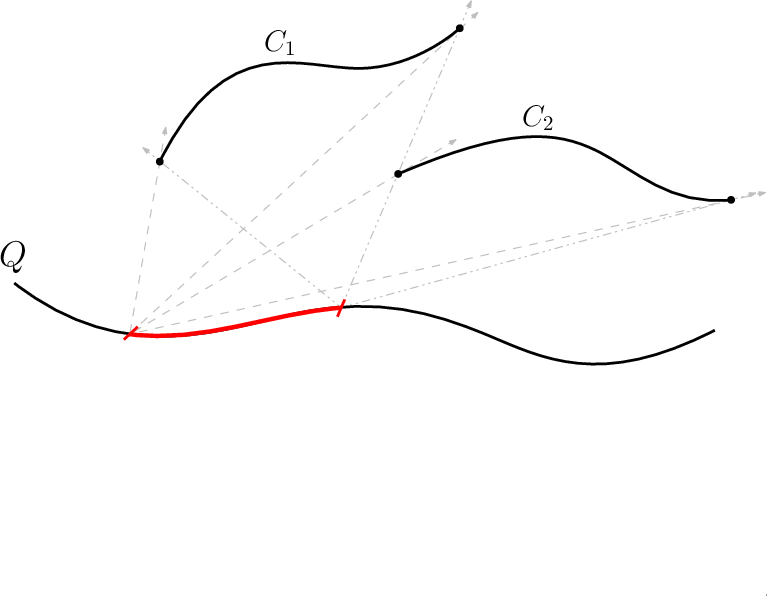}
    \caption{\small A configuration similar to that shown in Figure~\ref{fig:wedge} for two curves $C_1$ and $C_2$ is depicted, the highlighted segment being cyclically invariant with respect to the given population, as can be seen by inspecting the wedge boundaries}
    \label{fig:wedgeSegment}
\end{figure}

\begin{definition}[Cyclically Invariant Segments]\label{def:cyclicalInvaraince}
A segment along a curve that maintains the same cyclic ordering of boundaries within  the circular partitions of each point along its length, is called {\em cyclically invariant}. Specifically, for a given curve $Q$, a segment $I\subseteq Q$ is {\em cyclically invariant} provided $\sigma_q(\mathcal{W}_\mathcal{C})$ has the same ordering of wedge boundaries as $\sigma_{q'}(\mathcal{W}_\mathcal{C}')$, for all $\mathcal{W}_\mathcal{C}$ and $\mathcal{W_\mathcal{C}'}$ defined relative to any $q,q'\in I$ respectively. 
\end{definition} 

See Figure~\ref{fig:wedgeSegment}.
Clearly such segments exist when $\mathcal{P}\cup \{Q\}$ is a set of polylines in $\mathbb{R}^2$. This property might not hold more generally for all plane curves\footnote{We use $\mathcal{C}$ to denote a general set of plane curves, and $\mathcal{P}$ to denote a set of polylines in $\mathbb{R}^2$.}. For the remainder of this article, we assume $\mathcal{P}\cup \{Q\}$ is a set of polylines. 

\begin{lemma}[Invariant Segments along Polylines]\label{lem:invariantSegments}
Given a polyline $Q$ and a set $\cal P$ of polylines, $Q$ can be partitioned into line segments, each of which is cyclically invariant with respect to $\cal P$. 
\end{lemma}

\begin{proof}
Consider any line segment $I$ of $Q$, and assume without loss of generality that every polyline of $\mathcal{P}=\{P_1,P_2,\dots,P_n\}$ lies on one side the supporting line $L_I$ of $I$. An analogous argument can be applied to polylines that do not adhere to this condition, as any polyline that crosses $L_I$ can be partitioned into separate polylines that are on either side of $L_I$. Two tangent points, say $\tau_1$ and $\tau_2$ associated with $P_i$ and $P_j$ respectively, in a circular partition can only undergo a change in their cyclical ordering when the reference point (root) $q$ becomes collinear with one of the common tangents between the pair of polylines that define the associated wedges; that is, $\tau_1$ and $\tau_2$ may only swap relative positions when $q$ becomes collinear with an element of 
\begin{equation*}
    \tau(\{P_i,P_j\}) = \{\text{all lines } l \mid (l\cap \ch(P_i))\cup(l\cap \ch(P_j)) = \{p_{i,i'},p_{j,j'}\} \},
\end{equation*}
where $p_{i,i'}$ and $p_{j,j'}$ are points on $\partial\ch(P_i)$ and $\partial\ch(P_j)$ respectively. Consequently, as at most four such tangents exist for each pair of polylines, the set of points along $L_I$ that trigger change in wedge orderings must be finite. Therefore, $L_I$, and consequently $I$, can be partitioned into a finite number of cyclically invariant segments. Moreover, each such segment along $I$ is a maximal line segment on $Q$ between two consecutive points that trigger changes.
\end{proof}

By Observation~\ref{obs:stabbing} and Lemma~\ref{lem:invariantSegments}, the double integral in Eq.~\eqref{eq:depth} can be reformulated as a sum of integrals measuring the total angular area swept out by the wedges of $\mathcal{P}$ with stabbing number weights along all cyclically invariant segments. This reformulation, which is made explicit in Section~\ref{sec:algCalculation}, is possible due to the fact that stabbing numbers remain constant within circular partitions which intern remain unchanged along each invariant segment.

\section{Computing Curve Stabbing Depth for Polylines}
\label{sec:alg}
In this section we develop an algorithm for computing the exact curve stabbing depth of a given $m$-vertex polyline $Q$ relative to a given set $\cal P$ of $n$ polylines in the plane, each of which has $O(m)$ vertices. We assume that vertices in $\{Q\} \cup \mathcal{P}$ are in general position in $\mathbb{R}^2$.

\subsection{Algorithm Overview}
\label{sec:alg.overview}
The algorithm partitions $Q$ into cyclically invariant segments, computes the set of wedges and corresponding stabbing numbers for each cyclically invariant segment, each of which is integrated and summed along the length of $Q$. 


In particular, for each segment $Q_i$ of $Q$ and each polyline $P_j$ in $\mathcal{P}$, the algorithm constructs the convex hulls of the connected components of $P_j \setminus (P_j \cap Q_i)$. $Q_i$ is then partitioned into $O(m)$ cyclically invariant segments by identifying intersections between $Q_i$ and a line tangent to one of the convex hulls. This is repeated for each polyline in $\mathcal{P}$ to further partition $Q_i$ into $O(nm)$ cyclically invariant segments. The internal angles of the wedges rooted along each cyclically invariant segment are likewise determined by the polylines in $\mathcal{P}$. For each cyclically invariant segment, the internal angle and the number of polylines of $\mathcal{P}$ stabbed in each wedge relative to a point $q$ on $Q$ is maintained using a pseudoline arrangement associated with each convex hull. The curve stabbing depth along $Q_i$ is then calculated by integrating wedge angles along the sequence of cyclically invariant segments of $Q_i$.

\subsection{Preprocessing}\label{sec:algPreprocessing}
This section describes the data structures used for computing the exact curve stabbing depth of a polyline $Q$ relative to a given set $\mathcal{P}$ of polylines, where $Q=(q_1,q_2\dots,q_m) \in (\mathbb{R}^2)^m$ and $\mathcal{P}=\{P_1,P_2,\dots,P_n\}$, with $P_i=(p_{i1},p_{i2},\dots,p_{im}) \in (\mathbb{R}^2)^m$ for each $i \in \{1, \ldots, n\}$.

\subsubsection{Building The Hierarchy of Convex Hulls}
As the tangent points of a wedge $w(q,P_i)$ are dependent on the position of the root $q$ (which moves along the length of $Q$) we require a method for keeping track of the current tangent points which can quickly provide updates in response to changes in the position of $q$ on $Q$. A hierarchy of nested convex hulls built atop the polylines $\mathcal{P}$ can serve this purpose.

To construct this data structure, begin by computing the convex hull of each polyline in $\mathcal{P}$; let $\mathcal{H}$ denote the set of all such convex hulls. Having constructed $\mathcal{H}$, associate to each $\ch(P_i)\in \mathcal{H}$ the convex hulls formed from the subpaths of $P_i$ that result from cutting $P_i$ by the segments of $Q$. In this way, each polyline $P_i$ has associated to it an outermost convex hull $\ch(P_i)$ and a collection of convex hull sets totaling at most $O(m^2)$ elements, $\mathcal{H}(P_i)$, whose elements, denoted $\mathcal{H}(P_i,I)$, contain the convex hulls of subpaths of $P_i$ that result from $P_i$ being cut by a specific segment $I$ of $Q$.

\begin{figure}
    \centering
    \includegraphics[scale=0.8]{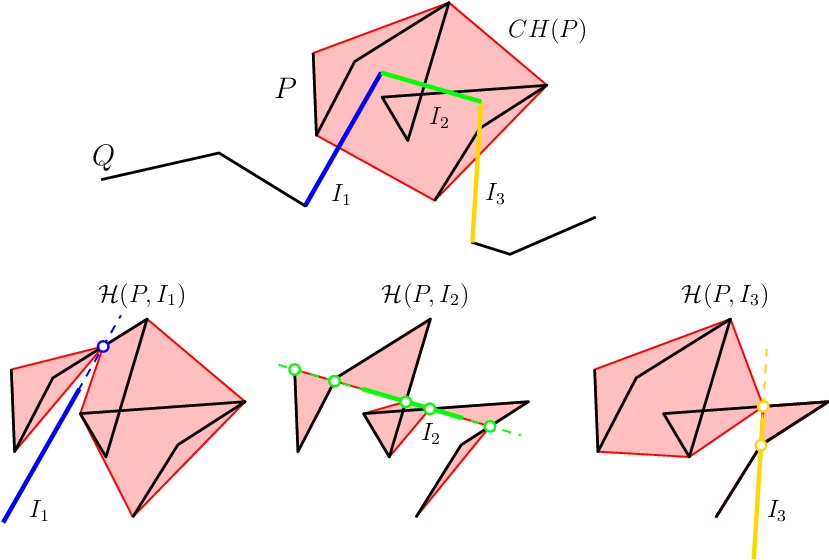}
    \caption{The hierarchy of convex hulls associated with a polyline $P$ as indexed by segments $I_1$, $I_2$, and $I_3$ of $Q$. Indicated in red are the convex hulls that encapsulate subpaths of $P$ generated by the intersecting segment of $Q$. Note the singular dangling segment formed by cutting $P$ by $I_3$ is considered a degenerate convex hull.}
    \label{fig:hierarchyConvexHulls}
\end{figure}

\begin{lemma}\label{lem:convexHullHierachy}
This hierarchy of convex hulls can be built in $O(nm^2\log m)$ worst-case time. 
\end{lemma}
\begin{proof}
Constructing $\mathcal{H}$ takes $O(nm\log m)$ time as each individual convex hull can be computed in $O(m\log m)$ time \cite{Chan:1996,Graham:1972} and $\mathcal{P}$ consists of $n$ polylines. Moreover, construction of the convex hulls subdivided by subpaths of polylines takes $O(nm^2\log m)$ additional time, on top of the $O(nm^2)$ time needed in the worst case to compute all the points of intersection. To see this, observe that $Q$ can intersect any $P_i\in\mathcal{P}$ at most $m^2$ times. Any particular segment of $Q$ contributes at most $m$ of these points, and results in at most the same number of subpaths of $P_i$ formed by this intersection. Let $m_1,m_2,\dots,m_m$ denote the number of vertices on each such subpath. We know $m_1+m_2+\dots + m_m  = 3m$, with $2m$ points resulting from double counting of division points. The time needed to build the convex hulls of these subpaths is given by 
\begin{equation*}
    \sum_{i=1}^m O(m_i\log m_i) \subseteq O(3m\log (3m)) = O(m\log m),
\end{equation*}
which follows from the inequality $a\log a + b\log b\leq (a+b)\log (a+b)$ for any $a,b\geq 0$. Thus, constructing the inner convex hulls, $\mathcal{H}(P_i,I)$ for each $I$, takes at most $O(nm^2\log m)$ time.
\end{proof}

\subsubsection{Dynamically Maintaining Wedge Tangent Points}
The primary method used to determine current tangent point pairs of a convex hull $\ch(P_i)$ (or $\ch(P_i,I)\in\mathcal{H}(P_i,I)$), for which $q\not\in\ch(P_i)$ (resp., $q\not\in\ch(P_i,I)$ for any $\ch(P_i,I)\in\mathcal{H}(P_i,I)$), is based on whether or not two paired curves, equal to the angles between $q$ and vertices of the convex hulls boundary as defined in Section~\ref{sec:wedgeAngles}, have the largest difference between them at point $q$. Such curves display several important properties necessary for the ensuing algorithm.

\begin{lemma}[Pseudolines]\label{lem:pseudolines}
The angular curves associated with each segment of a polyline (or vertex of a convex hull) relative to a point $q$ along a segment $I$ of $Q$, defined by arctangents (as seen in Section~\ref{sec:wedgeAngles}), are pseudolines, curves where any two intersect exactly once, when taken on the supporting line $L_I$ of $I$.
\end{lemma}
\begin{proof}
To begin, we must detail the form of the equations defining the curves in question. As we are partitioning a polyline $P_i$ on either side of segments of $Q$ all the curves either take one of the two possible forms discussed in Section~\ref{sec:wedgeAngles}. Without loss of generality, we may assume all the points are transformed so that the segment $I$ is collinear with the horizontal axis. From this, it follows the angular curves are of the form 
\begin{align*}
    \theta(q,p) = \frac{\pi}{2}-\tan^{-1}\left(\frac{p_x-q_x}{p_y}\right) \quad\text{ or }\quad
    \theta(q,p) = \pi+\frac{\pi}{2}-\tan^{-1}\left(\frac{p_x-q_x}{p_y}\right),
\end{align*}
depending on whether the transformed point $p=(p_x,p_y)$ lies above or below $I$. Thus, it suffices to consider only equations of the first form for the proof. It is obvious that two curves, determined by points $p_1$ and $p_2$ of this form intersect exactly at the point along $L_I$ corresponding to the intersection between $L_I$ and the extension of the segment $\overline{p_1p_2}$. Consequently, it follows that any pair of curves can intersect at most once due to the fact that the points defining these curves are selected from distinct segments along polylines in general position with unqiue subpaths; for clarity, two curves do not intersect if and only if the two points $p_1$ and $p_2$ defining them determine a line parallel to $L_I$, a condition that cannot be satisfied under the assumption of general position (see Section~\ref{sec:degenerateCases}). A similar argument holds for curves of the second variety.
\end{proof}

Continuing we see that these pseudolines obey a useful ordering principle.

\begin{lemma}[Ordering]\label{lem:pseudolineordering}
Any set of pseudolines $\mathcal{L}$ defined as in Lemma~\ref{lem:pseudolines} obeys the following ordering condition. Let $V_l$ be a vertical line to the left of the leftmost point of intersection among the elements of $\mathcal{L}$. The line $V_l$ induces a total order $\leq$ on $\mathcal{L}$, where for any $l_1,l_2\in \mathcal{L}$, the relation $l_1\leq l_2$ is true if and only if $l_1$ intersects $V_l$ below (with respect to $y$-coordinates) the intersection of $l_2$ with $V_l$. If $\mathcal{L}$ has the ordering $l_1\leq l_2\leq \cdots\leq l_n$ with respect to $V_l$, then the total order induced on $\mathcal{L}$ by a vertical line $V_r$ to the right of the rightmost point of intersection satisfies $l_1\geq l_2\geq \cdots\geq l_n$; that is, the ordering of the elements of $\mathcal{L}$ are exactly reversed.
\end{lemma}
\begin{proof}
As was done in the proof of Lemma~\ref{lem:pseudolines}, we only consider curves of the first form, let $\mathcal{L}$ be the resulting set of curves. By construction, each curve $l$ in $\mathcal{L}$ asymptotically ranges between $0$ and $\pi$ when taken on the respective supporting line, say $L_I$. Moreover, all curves in $\mathcal{L}$ are monotonically increasing with different slopes due to assumption of general position. Moreover, each pair of curves $l_1,l_2\in \mathcal{L}$ with $l_1\neq l_2$ cross exactly once. Consequently, it follows that the condition on ordering must be satisfied; otherwise, two curves that have not swapped order must not have intersected at least one of other curves, which would contradict the result of Lemma~\ref{lem:pseudolines}.
\end{proof}

Beyond this, we can find the point of intersection between any two curves of the same class in constant time. 

\begin{lemma}\label{lem:constantIntersection}
The point of intersection between any two angular curves defined in terms of the vertices of $\partial\ch(P)$ can be found in constant time.
\end{lemma}
\begin{proof}
As intersections are calculated solely between curves of the same class, it follows by simple algebraic manipulation that intersections can be found between curves of matching form on the finite segment of interest in finite time. For example, for two curves of the first kind 
\begin{align*}
    \theta(q,p_1) = \frac{\pi}{2}-\tan^{-1}\left(\frac{p_{1x}-q_x}{p_{1y}}\right) \quad\text{ and }\quad 
    \theta(q,p_2) = \frac{\pi}{2}-\tan^{-1}\left(\frac{p_{2x}-q_x}{p_{2y}}\right),
\end{align*}
with $q$ again assumed to be collinear with the horizontal axis. One can derive the point of intersection by solving
\begin{align*}
    \frac{p_{1x}-q_x}{p_{1y}} &= \frac{p_{2x}-q_x}{p_{2y}}
\end{align*}
for $x=q_x=(p_{1y}p_{2x}-p_{2y}p_{1x})/(p_{1y}-p_{2y})$ and $y=\tan^{-1}(x)$.
\end{proof}

The results of Lemmas~\ref{lem:pseudolines}, \ref{lem:pseudolineordering}, and~\ref{lem:constantIntersection} enable the use of the data structure presented in \cite{Agarwal:2019} to maintain the current tangent points of the wedges associated with each polyline, by building the upper and lower envelopes of the angular curves defined by the points on the convex hulls of $P$ relative to the position of $q$ along $Q$. More precisely, for each segment $I$ of $Q$, and for each $\ch(P_i)\in\mathcal{H}$, we construct two paired instances of the pseudoline data structure, with each pair responsible for maintaining matching upper and lower envelopes of sets of the aforementioned pseudolines. 

Tangent points for wedge $w(q,P)$ with $q\not\in\ch(P_i)$ (Case~1 for the cases outlined in Section~\ref{sec:defs}) are maintained by constructing one such data structure pair using the curves defined by vertices along $\partial\ch(P_i)$. This enables tangent points to be enumerated on the upper and lower envelopes simultaneously, by maintaining a reference to the aforementioned vertices of $\partial\ch(P_i)$ that define the segments appearing on each.

Complications arise when attempting to maintain the tangent points of $w(q,P_i)$ with $q\in\ch(P_i)$, that is, for Case~2, using this method. A solution is to instead examine tangent points for \emph{complementary wedges}, denoted $w'(q,P_i)$, which are wedges defined by segments of $\partial \ch(P_i)$ that yield $w(q,P_i)$ after taking their difference with the plane. In particular, the vertices of \emph{window segments} of $\ch(P_i)$, edges of $\partial \ch(P_i)$ that do not belong to $P_i$, induce such complementary wedges. These window segments can be easily identified during the computation of $\mathcal{H}$. It is also important to note that any planar face of $\ch(P_i)$, formed by the planar subdividion of $\mathcal{H}(P_i,I)$, can include at most one window segment on its boundary. This allows for the unique determination of wedge tangent points, as we will see.

In order to dynamically compute these complementary wedges and later refine them to account for occlusions within $\ch(P_i)$, we require the use of a second pair of pseudoline data structures built using the points that define the boundaries of the convex hulls of $\mathcal{H}(P_i,I)$. In particular, tangent points for window wedges are extracted in the following fashion. For any $q$ along $I$ within $\ch(P_i)$, we first identify the planar face in which $q$ lies based on the last point of intersection between $I$ and $\partial\ch(P_i)$ or $\mathcal{H}(P_i,I)$ before $q$. From this, the corresponding window edge, if one exists, can be immediately determined due to each planar face being associated with at most one such segment. Then, as is done for Case~1, utilizing the upper and lower envelopes of the pseduoline data structures of elements in $\mathcal{H}(P_i,I)$ adjacent to the window edge we can derive the unobstructed internal angle of the wedge $w'(q,P_i)$, along with corresponding potential tangent points, say $\tau_1$ and $\tau_2$. However, unlike in Case~1, additional checks need to be performed to ensure the upper and lower curves define a valid tangent point selection. Namely, the tangent points $\tau_1$ and $\tau_2$ are further refined to account for obstructions within the inscribed wedge $w'(q,P_i)$ by inspecting the vertices of $\mathcal{H}(P_i,I)$ that lie within the wedge. These vertices can be easily identified, again, using the pseduoline data structure associated with the elements of $\mathcal{H}(P_i,I)$ and checking if the angles associated with a vertex falls within $w'(q,P_i)$. These obstructing vertices, if any are found, are then used to restrict the internal wedge. This, is accomplished by comparatively sorting their angles with respect to those of $\tau_1$ and $\tau_2$.

Beyond this, visibility to a window segment can be lost completely within a convex hull when the two curves representing the current tangent points of the window wedge cross. Moreover, there does not exist a valid selection of tangent points for a complimentary wedge $w'(q,P_i)$ when $q$ is contained within two nested convex hulls of $\mathcal{H}(P_i,I)$ that appear on opposing sides of segment $I$. Such a configuration can be identified by preemptively sorting the left and right points of intersection of elements of $\mathcal{H}(P_i,I)$ on $I$, their respective inducing segments of $Q$, and augmenting the sorted points to include which ``side'' of the segment they lie on; due to the looping of these curves, the typical notion of right and left will not be consistent; instead we assign an arbitrary right-left orientation at the start of a curve, which is maintained as we sweep along its length.

\begin{lemma}\label{lem:pseudolineDataStructure}
The set of all matched pair pseudoline data structures responsible for maintaining the upper and lower envelopes of tangent point curves can be built in $O(nm^2\log^2 m)$ time.
\end{lemma} 
\begin{proof}
A pseudoline data structure is built for every $\ch(P_i)\in\mathcal{H}$, each of which consists of at most $m$ pseudolines that are defined in terms of  the vertices of $\partial\ch(P_i)$. The data structure is capable of inserting and merging each new entry in $O(\log^2m)$ time, for $m$ currently stored elements. Thus, for $m$ insertions the data structure requires 
\begin{equation*}
    \sum_{i=1}^{m}O(\log^2 i) \subseteq O(m \log^2 m) \text{ time.}
\end{equation*}
Consequently, the first matched pair of data structures can be built in $O(nm\log^2 m)$ worst-case time.
Likewise, each segment $I$ of $Q$ can give rise to at most $m$ internal convex hulls formed by partitioning $P_i$, with at most $2m$ total vertices appearing on the boundaries of the convex hulls of $\mathcal{H}(P_i,I)$. Thus, at most $2m$ entries are added to the pseudoline data structure associated with each $\mathcal{H}(P_i,I)$. Similar analysis to the first case concludes with all upper and lower envelopes of the inner convex hulls being constructed in $O(nm^2\log^2 m)$ worst-case time. Thus, $O(nm^2\log^2 m)$ time is required in the worst case to build all instances of the data structure.
\end{proof}

\begin{lemma}\label{lem:convexHullSorting}
Sorting overlapping nested convex hulls takes $O(nm^2\log m)$ total time.
\end{lemma}
\begin{proof}
Each $\ch(P_i)$ admits at most $m$ different sets of $\mathcal{H}(P_i,I)$, one for each $I$ of $Q$, each with $O(m)$ points of intersection that need to be sorted along $I$. It follows for all $n$ polylines in $\mathcal{P}$ it takes $O(nm^2\log m)$ time to sort all $\mathcal{H}(P_i,I)$.
\end{proof}

As we will be effectively traversing the entire upper and lower envelopes associated with each segment of $Q$, it is more efficient in the worst case to preemptively extract the complete list of segments and intersection points of all pseudolines appearing on the upper and lower envelopes opposed to performing a series of queries. Since the utilized data structure is based on a balanced binary search tree, it is possible to traverse the tree and construct the segment list and associated transition points as desired. Let $T(P_i)$ denote the sequence of pseudonlines and their points of intersection for the pseudoline data structure associated with $\ch(P_i)$, and likewise, let $T(P_i,I)$ denote the sequences of pseudolines associated with $\mathcal{H}(P_i,I)$. 

\begin{lemma}\label{lem:pseudolineExtraction}
The sequence of segments and points of intersection in $T(P_i,I)$ (including $T(P_i)$) can each be extracted from the associated pseudoline data structure in $O(m)$ time. It takes $O(nm^2)$ time to extract all such sequences across all pseudoline data structure instances.
\end{lemma}
\begin{proof}
As the data structure used for maintain the lower (and upper) envelope of pseudolines is implemented using a balanced binary tree \cite{Agarwal:2019}, we can traverse the tree in $O(m)$ time and extract the sequence of segments that appear on the lower (upper) envelope for each instance of the data structure. Consequently, as there are $O(nm)$ instances, it takes $O(nm^2)$ time to extract all such sequences.
\end{proof}

By the above proofs, it follows that the set $T$, which contains all points in both $T(P_i)$ and $T(P_i,I)$, has cardinality $O(nm^2)$. This will be important in the later analysis of the running time of the wedge updates. 

Having described the method used for maintaining tangent point information, we briefly discuss its operation in the following algorithm. The initial tangent points $\tau_1$ and $\tau_2$ of $w(q,P_i)$ for all $i=1,\dots,n$ and $q=q_1$, the first point along $Q$, are derived as described in the preceding paragraphs and arranged into the circular partition $\sigma_q(\mathcal{W}_\mathcal{P})=(\theta_0,\theta_1,\dots,\theta_{4n})$ by sorting the rays associated to each tangent point by slope, treating the opposing ray extended through the origin separately. During this process, note which regions overlap to calculate the initial stabbing numbers of each angular region in the partition (subdivided wedges) as in Eq.~\eqref{eq:stabObs} and Observation~\ref{obs:stabbing}. This ordering is then dynamically maintained as tangent points are incrementally updated as $q$ moves along $Q$ by tracing the corresponding upper and lower envelopes. Similarly, the stabbing numbers are iteratively updated by monitoring the points of $S$ defined along $Q$, determined using a method now described.

\subsubsection{Dynamically Maintaining Stabbing Numbers}

A method capable of first calculating and then maintaining the stabbing number associated with each wedge is required. This is again accomplished using the set of convex hulls, $\mathcal{H}$, by constructing the set $\tau(\mathcal{H})$ of common tangent lines that separate all pairs of convex hulls, as utilized in the proof of Lemma~\ref{lem:invariantSegments}, to partition $Q$ into segments with fixed stabbing numbers. Specifically,   
\begin{equation*}
    \tau(\mathcal{H}) = \{\text{all lines } l \mid \exists \{P_i,P_j\}\subseteq \mathcal{P} \ \text{s.t.} \ (l\cap \ch(P_i))\cup(l\cap \ch(P_j)) = \{p_{i,i'},p_{j,j'}\} \},
\end{equation*}
where $p_{i,i'}$ and $p_{j,j'}$ are vertices of $\partial\ch(P_i)$ and $\partial\ch(P_j)$, respectively. 
\begin{figure}[!ht]
    \centering
    \includegraphics[scale=0.6]{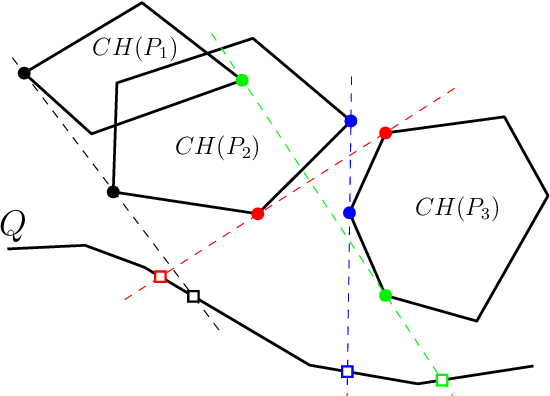}
    \caption{\small Illustration of the common tangents that separate pairs of the convex hulls $\ch(P_1),\ch(P_2)$, and $\ch(P_3)$. To simplify the figure, only those tangents that intersect $Q$ are shown, with their points of intersection marked along $Q$ by boxes.}
    \label{fig:hullTangents}
\end{figure}

\begin{lemma}\label{lem:commonTangents}
The set of all common tangents between all elements of $\mathcal{H}$ can be computed in $O(n^2m\log^2 m)$ time.
\end{lemma}
\begin{proof}
There are three distinct cases to consider when computing these common tangents: (1) the two convex hulls are disjoint, (2) their boundaries intersect, and (3) one convex hull entirely contains the other. Case~1 is the simplest, in which the common tangents between two convex hulls $\ch(P_1)$ and $\ch(P_2)$ can be computed in $O(\log|\partial\ch(P_1)|+\log|\partial\ch(P_2)|)$ time \cite{Kirkpatrick:1995}. Case~2 requires $O(m)$ time to compute in the worst case. However, if the two convex hull boundaries intersect at most twice, the common tangents can be found in $O(\log(|\partial\ch(P_1)|+|\partial\ch(P_2)|)\log k)$ time, where $k=\min\{|\partial(\ch(P_1)\cap\ch(P_2))|,|\partial(\ch(P_1)\cup\ch(P_2))|\}$ \cite{Kirkpatrick:1995}. In Case~3, no computation is performed after identifying that the hulls are nested. It takes $O(m)$ time to identify which of the three cases must be applied \cite{ORourke:1982}.
\end{proof}

Marking the points of intersection between lines of $\tau(\mathcal{H})$ and $Q$ yields a point set, $S$, that identifies when wedge stabbing numbers need to be updated relative to the position of $q$ along $Q$; see Figure~\ref{fig:hullTangents}. Wedge stabbing numbers are then iteratively updated as $q$ crosses points of $S$ along $Q$ by updating the current tangent point ordering (wedge overlap) information to reflect the change in overlap between the convex hulls whose common tangent defines the crossed point. 

\begin{lemma}\label{lem:stabbingNumbers}
The initial stabbing numbers associated with all wedges can be computed in $O(n^2m)$ time and maintained in $O(n^2)$ additional time.
\end{lemma}
\begin{proof}
The initial tangent points can be computed and cyclically ordered in $O(n\log m + n\log n)$ time, with the initial stabbing numbers being computed during the sorting stage based on overlap determined by tangent point crossings. There are $4 \binom{n}{2}$ common tangents in total, each of which must be tested for intersection with segments of $Q$. Correspondingly, in the worst case each tangent intersects $Q$, and intersection testing takes $O(n^2m)$ time. Maintenance of stabbing numbers then requires at most $O(n^2)$ discrete constant-time updates along the length of $Q$.
\end{proof}
     
We can conclude that it is possible to build a set of data structures capable of progressively updating the tangent points and stabbing numbers associated with all wedges in $O(n^2m\log^2 m + nm^2\log^2m)$ total time, which is sufficient for later calculating the curve stabbing depth of a polyline.

\subsection{Calculating the Depth}\label{sec:algCalculation}     
We can now calculate the curve stabbing depth of a polyline, by application of the data structures constructed in Section~\ref{sec:algPreprocessing} to evaluate wedge angles by applying the equations outlined in Section~\ref{sec:wedgeAngles}, which in Case~1 are given by Equations~\ref{eq:interiorAngleAbove} and~\ref{eq:interiorAngleCross}.

The depth of a curve $Q$ is computed by progressively summing the depths of segments along $Q$ as $q$ transitions across its length. Specifically, the angular area of each wedge given by the difference between the associated upper and lower envelopes is integrated along the segments of $Q$ that result from subdividing $Q$ where tangent points and stabbing numbers are updated. 

Let $\mathcal{I}$ denote the partition of $Q$ into cyclically invariant segments by points of $S$. Then given the set $\mathcal{W}_\mathcal{P}$ of initial wedges as described above, the cyclical invariance of $I\in\mathcal{I}$ allows the angular area swept out by a wedge, $w(q,P)$, along each subsegment $I_i=\overline{ab}$ of $I$ formed between tangent update points of $T$ to be evaluated as
\begin{equation}
    A_i = \int_{q\in I_i} \theta(w(q,P))\, ds,
    \label{eq:unparamInt}
\end{equation}
where $\theta(w(q,P))$ may be one of \ref{eq:interiorAngleAbove} or~\ref{eq:interiorAngleCross} (or similarly for Case~2) as outlined in Section~\ref{sec:wedgeAngles}.

\begin{lemma}\label{lem:tangnetPoints}
We can compute all changes in wedge tangent points, including losses in visibility within the interior of a convex hull, and calculate wedge angles in $O(nm^2)$ total time as $q$ transitions across the length of $Q$.
\end{lemma}
\begin{proof}
We need to extract the tangent point information at each crossing on the upper and lower envelopes. There are $O(nm)$ many of these associated with each segment of $Q$ from all polylines, so this takes $O(nm^2)$ time in total. 
\end{proof}


By letting $\vec{u}=(u_x,u_y)$ denote the unit direction vector inline with the line segment $\overline{q_iq_{i+1}}$ which $I$ belongs, we can construct a $3$-by-$3$ affine transformation matrix of the form
\begin{equation*}
    M = \left[\begin{array}{c|c}
            R & T\\ \hline
            0 & 1
        \end{array}\right],
\end{equation*}
composed of a transition (rotation) matrix $R$ that is responsible for rotating $\vec{u}$ to be horizontal, and a translation matrix $T$ that displaces the newly rotated $I$ to a height of zero. In particular, we see 
\begin{equation*}
    R = \begin{bmatrix}
            u_x & u_y\\
            -u_y & u_x
        \end{bmatrix}
    \quad\text{ and }\quad
    T = \begin{bmatrix}
            -q_{i,x}\\ -q_{i,y}
        \end{bmatrix}.
\end{equation*}
Applying this transformation to the set of current wedge tangent points, that is, if $\tau=(\tau_x,\tau_y)$ is a tangent point then the resultant transformed point $\tau'=(\tau'_x,\tau'_y)$ is given by $(\tau'_x,\tau'_y,1)=M(\tau_x,\tau_y,1)$ with a one appended to the point vector in the calculation, for $I\in\mathcal{I}$ allows us to calculate the area swept out by wedges along $I$ in a unified fashion; e.g., for Case~1(a) using Eq.~\eqref{eq:interiorAngleAbove} in the integral, results in Eq.~\eqref{eq:unparamInt} becoming

\begin{align*}
\small
    A_i = \int_{a'}^{b'}
    \bigg[&\tan^{-1}\left(\frac{x-\tau_2(w')_x}{\tau_2(w')_y}\right)
    -\tan^{-1}\left(\frac{x-\tau_1(w')_x}{\tau_1(w')_y}\right)\bigg]\, dx, \label{eq:wedgeArea}
\end{align*}
with the transformed points $a',b'$ and wedge $w'$ defined by the tangent points delineating $I_i$. As this is an integral with known antiderivative, namely 
\begin{align*}
\small
    A_i =& \bigg[ (\tau_1(w')_x-x)\tan^{-1}\left(\frac{\tau_2(w')_x-x}{\tau_2(w')_y}\right)\\
    &+ (x-\tau_1(w')_x)\tan^{-1}(\tau_1(w')_y(\tau_1(w')_x-x))\\
    &+ \frac{1}{2\tau_1(w')_y}\ln(\tau_1(w')_y(\tau_1(w')_x^2-2\tau_1(w')_xx+x^2)+1)\\
    &- \frac{1}{2}\tau_2(w')_y\ln(\tau_2(w')_x^2-2\tau_2(w')_xx + \tau_2(w')_y^2 + x^2) \bigg]_{a'}^{b'},
\end{align*}
the angular area, and therefore depth, can be computed exactly along such segments. Analogous analysis can be applied using Eq.~\eqref{eq:interiorAngleCross} for problems in Case~1(b) reassembling that depicted in Figure~\ref{fig:wedgeAngleCross}, and likewise for computations of Case~2. 

\begin{lemma}\label{lem:wedgeArea}
The angular area of all wedges along a cyclically invariant segment of $Q$ can be computed in $O(nm)$ worst-case time by applying a coordinate transform along the segment.
\end{lemma}
\begin{proof}
Along each invariant segment $I$ of $Q$, wedges, as derived per Lemma~\ref{lem:tangnetPoints}, can undergo at most $O(m)$ tangent point updates. Thus, we need only transform $O(m)$ points (pseudolines) to maintain all tangent points of a wedge. Consequently, as there are $O(n)$ wedges, the total angular area of all wedges along $I$ can be computed in $O(nm)$ time. 
\end{proof}

For each $I\in \mathcal{I}$, we select a minimizing subset of the circular partition, $\sigma_q(\mathcal{W}_\mathcal{P})$, by defining the indicator function (bit sequence)
\begin{equation*}
\small
    \mathbb{I}_i = 
    \begin{cases}
        1 & \text{ if } \stab_{\mathcal{P}}(\overrightarrow{q_{\theta^*}}) \leq \stab_{\mathcal{P}}(\overrightarrow{q_{\theta^*+\pi}}) \text{ for all } \theta^* \in [\theta_{i-1}, \theta_i) \\
        0 & \text{ otherwise}, 
    \end{cases}
\end{equation*}
for $i = 1,\dots,4n$. This selection procedure performs the same task as the minimization operation within Eq.~\eqref{eq:depth}. The initial values of $\mathbb{I}_i$ for $i=1,\dots,4n$ are found while sorting $\sigma_q(\mathcal{W}_\mathcal{C})$, and are then iteratively updated using the points of $S$ which we recall form the endpoints of element in $\mathcal{I}$.

\begin{lemma}\label{lem:indicatorVariables}
The initial indicator values of $\mathbb{I}_i$ for $i=1,\dots,4n$ can be computed in $O(n)$ worst case time, and then iteratively updated using an additional $O(n^2)$ time. 
\end{lemma}
\begin{proof}
The initial indicator values are computed by comparing the stabbing numbers of opposing wedges. As there are $O(n)$ many wedges, this procedure can be seen to take $O(n)$ worst case time. At each point of $S$ at most one wedge in $\sigma_q(\mathcal{W}_\mathcal{C})$ can have its stabbing number updated, and consequently as each such wedge is formed from opposing subdivisions at most two values in $\mathbb{I}_i$ for $i=1,\dots,4n$ need to be updated. Then, as $S$ contains $O(n^2)$ points, it follows an additional $O(n^2)$ time is needed for maintaining the minimizing selection along all of $Q$. 
\end{proof}

The final depth accumulated along $I_i$ is given by the reformulation of Eq.~\eqref{eq:depth}
\begin{equation*}
\small
    D_{I_i} = \frac{1}{\pi L(Q)}\sum_{j=1}^{4n} \mathbb{I}_j\cdot \stab_{\mathcal{P'}}(\overrightarrow{q_{\theta^*_j}}) A_j, \label{eq:wedgeSum}
\end{equation*}
where $q$ is any point along $I_i$, sample angle $\theta^*_j \in [\theta_{j-1},\theta_j)$, and $A_j$ is the angular area swept out by the wedge $w(q)$ bounded between the angles $[\phi(\tau_1(w(q)),\phi(\tau_2(w(q))]=[\theta_{j-1},\theta_j]$ as $q$ is translated across $I_i$ as calculated using Eq.~\eqref{eq:unparamInt}. 

The total depth of $Q$ is then found by evaluating the sum 
\begin{equation*}
    D(Q,\mathcal{P}) = \sum_{I\in \mathcal{I}}\sum_{I_i\in I}D_{I_i},
\end{equation*}
with respect to changes in stabbing numbers between each subsegment of $I$ and tangent point updates of each $T(P,I)$ along $I$.

\begin{lemma}\label{lem:computingTotalDepth}
All computations necessary for computing the final curve stabbing depth of a polyline $Q$ with respect to a set of polylines $\mathcal{P}$ take $O(nm^2+n^3)$ time in the worst case.
\end{lemma}
\begin{proof}
This lemma follows directly from those given in Lemmas~\ref{lem:tangnetPoints}, \ref{lem:wedgeArea}, and \ref{lem:indicatorVariables}.
\end{proof}

This gives our main result:

\begin{theorem}[Computing Curve Stabbing Depth for Polylines]
\label{thm:mainAlg}
The curve stabbing depth of an $m$-segment polyline relative to a set of $n$ polylines, each with $O(m)$ segments, can be computed in $O(n^3 + n^2m\log^2 m + nm^2\log^2 m)$ time using $O(n^2+nm^2)$ space. 
\end{theorem}
\begin{proof}
This result follows immediately by combining the results of  Lemmas~\ref{lem:convexHullHierachy} and \ref{lem:pseudolineDataStructure}--\ref{lem:wedgeArea}, which detail the time required for prepossessing and subsequent computation stages of the algorithm.
\end{proof}

When $n \in \Theta(m)$, the running time in Theorem~\ref{thm:mainAlg} can be expressed as $O(r^{3/2}\log^2 r)$, where $r \in \Theta(nm)$ denotes the input size, i.e., the total number of vertices in 
$\{Q\} \cup \mathcal{P}$.

\subsection{Degenerate Cases: Relaxing the General Position Assumption}
\label{sec:degenerateCases}
If we relax the assumptions of general position to allow parallel segments between polylines of $\mathcal{P}$ and $Q$ Lemma~\ref{lem:pseudolineordering} no longer holds, and consequently the pseudoline data structure can no longer be directly used. In such situations, it is possible to subdivide the domains of the angular curves such that no two curves cross more than once on a single interval, and then apply the described method on each interval separately.

\section{Properties}
\label{sec:properties}
Curve stabbing depth is intended to serve as an empirical tool for the study of observational curve (trajectory) data without restriction on monotonicity. In this section we examine properties of this new depth measure.

In a fashion similar to the much cited work \cite{Zuo:2000}, which lists a set four properties multivariate depth functions should satisfy, the papers \cite{Reyes:2016} and \cite{Gijbels:2017} consider a list of six potentially desirable properties that functional depth measures ought to satisfy; the discussion therein highlights that there is currently no single generally assumed set of properties imposed on functional depth measures. Perhaps unsurprisingly, the domain of unparameterized curve depth measures, for curves which can be thought of informally as extending functional data, is equally sparse in established properties that ought to be satisfied. Despite this lack of agreement, we highlight a selection of these properties for our depth measure. 

Throughout Section~\ref{sec:properties}, the depth of a curve is taken to be normalized with respect to the number of curves in the sample; i.e., for a plane curve $Q\in \Gamma$ and a set $\mathcal{C}\subseteq \Gamma$ of $n$ plane curves, take 
\begin{equation*}
    D(Q,\mathcal{C})=\frac{1}{n \pi L(Q)}\int_{q\in Q}\int_{0}^{\pi} \min\{\stab_{\mathcal{C}}(\overrightarrow{q_\theta}),\stab_{\mathcal{C}}(\overrightarrow{q_{\theta+\pi}})\} \,d\theta\, ds.
\end{equation*}
This is done to eliminate dependence on the number of curves in $\mathcal{C}$ in the ensuing analysis.

\subsection{Bounded Depth}
It is trivial to observe the maximum obtainable stabbing number of a ray is bounded above by the number of curves in $\mathcal{C}$, and therefore the modified curve stabbing depth is bounded to the unit interval. More formally, we see
\begin{observation}
For any $Q$ and any finite subset $\mathcal{C}$ from $\Gamma$, the curve stabbing depth $D(Q,\mathcal{C})$ is bounded within $[0,1]$.
\end{observation}
\begin{proof}
By definition it follows that 
\begin{align*}
    0\leq D(Q,\mathcal{C}) &= \frac{1}{n \pi L(Q)}\int_{q\in Q}\int_{0}^{\pi} \min\{\stab_{\mathcal{C}}(\overrightarrow{q_\theta}),\stab_{\mathcal{C}}(\overrightarrow{q_{\theta+\pi}})\} \,d\theta\, ds\\
    &\leq \frac{1}{n \pi L(Q)}\int_{q\in Q}\int_{0}^{\pi} n \,d\theta\, ds\\
    &= 1.
\end{align*}
Where the lower bound of zero results from the fact that $\stab_{\mathcal{C}}(\cdot) \geq 0$, as shown in Observation~\ref{obs:non-zeroRegion}.
\end{proof}

\subsection{Nondegeneracy}
\label{sec:nondegeneracy}
Nondegeneracy is an important property when the goal of a depth measure is to provide a total ordering (ranking) to a set of elements, such as a center-outward ordering, and for the notion of a median to be well defined, as it establishes the supremum and infimum are distinguishable under the measure. If a depth measure proves to be degenerate, the comparison between two elements' depths is not necessarily meaningful.

\begin{corollary}\label{lem:nondegeneracy}
Curve stabbing depth is nondegenerate, meaning for any given finite set $\mathcal{C}\subseteq \Gamma$ of plane curves 
\begin{equation*}
    \inf_{Q\in \Gamma} D(Q,\mathcal{C}) < \sup_{Q\in \Gamma} D(Q,\mathcal{C}),
\end{equation*}
when $\ch(\mathcal{C})$, the convex hull composed of all curves in $\mathcal{C}$, has non-zero area. 
\end{corollary}
This corollary follows directly from the results of Section~\ref{sec:medianPoints} via Observation~\ref{obs:medianPoint}, which shows $0< \sup_{Q} D(Q,\mathcal{C})$, and Corollary~\ref{cor:vanishingInfinity}, which shows $\inf_{Q}D(Q,\mathcal{C}) = 0$.

\subsection{Transformation Invariance}
Transforming data to a different coordinate systems during the process of data analysis is common practice for multivariate data. Thus, the sensitivity of curve stabbing depth with respect to common classes of transformations is of natural interest.

\subsubsection{Affine Invariance}
\label{sec:properties.affine}
A depth measure, $d$, is {\em affine invariant} if for all plane curves $Q$, all sets $\mathcal{C}$ of plane curves, and all affine transformations $f:\mathbb{R}^2\to\mathbb{R}^2$, $d(f(Q), f(\mathcal{C})) = d(Q,\mathcal{C})$. That is, the depth of $Q$ relative to $\mathcal{C}$ remains unchanged when any affine transformation $f$ is applied to both $Q$ and $\mathcal{C}$.

Curve stabbing depth does not satisfy general affine invariance, as it is not invariant under shear transformations. Figure~\ref{fig:affineCounter} contains a counterexample. Moreover, the same figure illustrates the relative rankings of curves with respect to their depth need not be preserved under shear transformations. Consider as the population, $\mathcal{C}$, all line segments drawn in the figure (all black and blue line segments together). After applying the reverse shear transformation than depicted, the relative ranks of the two central curves will change from the central leftmost segment being deeper than the central rightmost to the opposite relation.
Various other depth measures for points in $\mathbb{R}^d$, e.g., Oja depth \cite{Oja:1983} and integrated rank-weighted depth \cite{Durocher:2019}, whose calculation depends on area or angles are also not invariant under affine transformations.

\begin{figure}[!ht]
    \centering
    \includegraphics[scale=0.8]{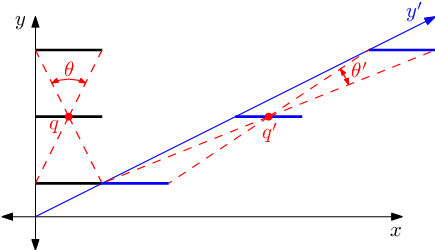}
    \caption{\small Three parallel line segments are drawn in the plane with a point $q$ and associated wedges highlighted in red on the center most segment. Drawn to the right, in blue, is the result of applying a shear transformation to the given segments, and highlighted wedge. In particular, notice the wedge angle is greatly reduced after the transformation, and thus the depth attributed to the point $q'$ is also proportionately reduced. This observation holds true for all points along the segments; consequently, the depth of all segments is not preserved.}
    \label{fig:affineCounter}
\end{figure}

\subsubsection{Similarity Invariance}
If a depth measure's calculation depends on areas or angles determined by its input, and it is not invariant under affine transformations, a natural property for it to satisfy is invariance under similarity transformations, a subset of affine transformations. 
Recall that two geometric objects are said to be {\em similar} if one can be obtained from the other by a similarity transformation consisting of a finite sequence of dilations and rigid motions (e.g., uniform scaling, translation, rotation, or reflection operations); see \cite{Modenov:1965} for a detailed discussion on geometric similarity. 

A depth measure, $d$, is invariant under similarity transformations if 
$d(f(Q),f(\mathcal{C})) = d(Q,\mathcal{C})$
for all sets of plane curves $\mathcal{C}$, all plane curves $Q$, 
and all similarity transformations $f:\mathbb{R}^2\to\mathbb{R}^2$.

As relative angles and lengths between points are preserved under similarity transformations, and the curve stabbing depth of a curve is normalized with respect to arc length and the sweep angle, it is immediate that curve stabbing depth satisfies the above criterion. The following observation is thus stated without proof.

\begin{observation}\label{lem:similarity}
Curve stabbing depth is invariant under similarity transformations.
\end{observation}

\subsection{Bounding the Region of Non-Zero Depth}
\label{sec:properties.non-zeroDepth}
Due to the discrete nature of stabbing rays, we can bound the region of non-zero depth for curves of  non-zero length using, $\ch(\mathcal{C})$, the convex hull of all planar curves in the sample population. Which is to say, only the portion of a curve that falls into $\ch(\mathcal{C})$ might contribute positively to the total depth score of the curve, with the proportion falling outside this region accumulating zero depth. This result is formalized in Observation~\ref{obs:non-zeroRegion}.

\begin{figure}[!ht]
    \centering
    \includegraphics[scale=0.55]{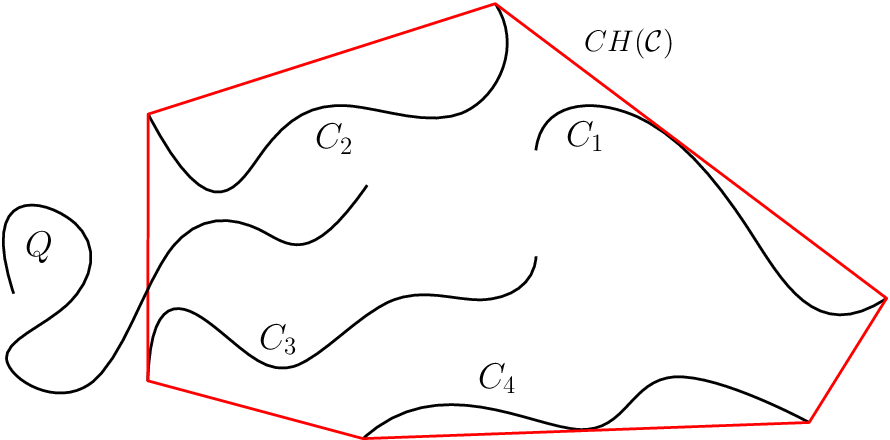}
    \caption{Depiction of the convex hull, $\ch(\mathcal{C})$, of a set of plane curves and a query curve $Q$ whose depth is being queried.}
    \label{fig:vanshingConvexHull}
\end{figure}

\begin{observation}\label{obs:non-zeroRegion}
For any set $\mathcal{C} \subseteq \Gamma$ of plane curves and any plane curve $Q \in \Gamma$, 
\begin{equation*}
    D(Q, \mathcal{C}) = \frac{L(Q \cap \ch(\mathcal{C}))}{L(Q)} D (Q \cap \ch(\mathcal{C}), \cal C).
\end{equation*}
 
%
\end{observation}
\begin{proof}
The integral definition of the curve stabbing depth of any such curve $Q$ with respect to $\mathcal{C}$ can be broken in two,
\begin{align*}
    D(Q,\mathcal{C}) = \frac{1}{n \pi L(Q)}\bigg[ &\int_{Q\cap \ch(\mathcal{C})}\int_{0}^{\pi} \min\{\stab_{\mathcal{C}}(\overrightarrow{q_\theta}),\stab_{\mathcal{C}}(\overrightarrow{q_{\theta+\pi}})\} \,d\theta\, ds\\ +& \int_{Q\cap \overline{\ch(\mathcal{C})}}\int_{0}^{\pi} \min\{\stab_{\mathcal{C}}(\overrightarrow{q_\theta}),\stab_{\mathcal{C}}(\overrightarrow{q_{\theta+\pi}})\} \,d\theta\, ds \bigg].
\end{align*}
With the first integral taken over the region of intersection between $Q$ and $\ch(\mathcal{C})$ and the second the portion of $Q$ that does not cross into $\ch(\mathcal{C})$. It is plain that 
\begin{equation*}
    \min\{\stab_{\mathcal{C}}(\overrightarrow{q_\theta}),\stab_{\mathcal{C}}(\overrightarrow{q_{\theta+\pi}})\}=0
\end{equation*}
for any $q\in Q\cap \overline{\ch(\mathcal{C})}$, since regardless of the particular value of $\theta$ only one ray may be oriented in the direction of $\ch(\mathcal{C})$ with the opposing ray avoiding all possible intersections. To do otherwise, would require nonconvexity of $\ch(\mathcal{C})$. Thus, the above integral decomposition reduces to 
\begin{equation*}
     D(Q,\mathcal{C}) = \frac{1}{n \pi L(Q)}\int_{Q\cap \ch(\mathcal{C})}\int_{0}^{\pi} \min\{\stab_{\mathcal{C}}(\overrightarrow{q_\theta}),\stab_{\mathcal{C}}(\overrightarrow{q_{\theta+\pi}})\} \,d\theta\, ds,
\end{equation*}
which is to say the portion of $Q$ outside of $\ch(\mathcal{C})$ incurs zero depth. Moreover, provided the conditions of Lemma~\ref{lem:nondegeneracy} are satisfied, it follows by the results of said lemma the last integral expression is necessarily non-zero.
\end{proof}

A corollary of Observation~\ref{obs:non-zeroRegion} is that any curve $Q$ that lies entirely outside $CH(\mathcal{C})$ has curve stabbing depth zero relative to $\mathcal{C}$. This is consistent with many measures of depth for point data, e.g., simplicial depth, Tukey depth, convex hull peeling depth, etc., where any query point outside the convex hull of the population has depth zero relative to the population. 

\subsection{Vanishing at Infinity}
\label{sec:vanishing}
Intuitively, the further away a curve is from the central cloud of curves, which in our case admits a convex hull boundary, the lower its depth score should be. Formulating such a property statement for our intensely geometric measure of depth is almost immediate. Specifically, a depth measure, $d$, is said to {\em vanish at infinity} if for any $Q\in \Gamma$ and a fixed finite subset $\mathcal{C}\subset\Gamma$,
\begin{equation*}
    \lim_{\dist(Q,\mathcal{C}) \to \infty}d(Q,\mathcal{C}) = 0 
\end{equation*}
as $Q$ is translated away from $\mathcal{C}$, where $\dist(Q,\mathcal{C})$ understood to represent a distance between the curve $Q$ and a sample of curves, $\mathcal{C}$, e.g., taking the minimum over all Fr\'{e}chet-distances between $Q$ and the elements of $\mathcal{C}$. 

\begin{corollary}\label{cor:vanishingInfinity}
For any sample of plane curves $\mathcal{C}$, the curve stabbing depth $D(Q,\mathcal{C})$ vanishes at infinity.
\end{corollary}

The proof of this corollary follows directly from the proof of Observation~\ref{obs:non-zeroRegion}. As such, only an outline is given.

\begin{proof}
Let $Q$ be a given plane curve, $\mathcal{C}$ be any set of plane curves, and $\ch(\mathcal{C})$ the convex hull containing all elements of $\mathcal{C}$. By Definition~\ref{def:curveDepth}, $D(Q,\mathcal{C})=0$ when $Q$ is translated completely outside of $\ch(\mathcal{C})$. Consequently, it is clear that $\lim_{\dist(Q,\mathcal{C})\to \infty}D(Q,\mathcal{C})= 0$, as desired.
\end{proof}

\subsection{Median Curves and Depth Median Points}
\label{sec:medianPoints}
The notion of a median element (or deepest element) under a depth measure is useful for all the same reasons as in standard multivariate settings, .e.g., providing a sample estimator or sample statistic of an unknown distribution. As a result, the topic of a median element appears more than once in the discussion of the next few sections. However, before such development can take place, we must define the form of a median element under our depth measure.

As a starting point, observe not all points along the length of a curve $Q$ contribute equally to the curve stabbing depth of $Q$ relative to the set of curves $\mathcal{C}$.  Recall, the depth of a point (a degenerate curve) is given by Eq.~\eqref{eq:pointDepth}. Points along $Q$ might individually score any depth value within $[0,1]$, with any points of $Q$ that fall outside of $\ch(\mathcal{C})$ condemned to a value of zero as shown by Observation~\ref{obs:non-zeroRegion}. Thus, it should be evident that for some point $q\in Q$, $D(q,\mathcal{C}) \geq D(Q, \mathcal{C})$. Consequently, we state the following observation.

\begin{observation}[Depth Median Points]\label{obs:medianPoint}
For any given set $\mathcal{C}$ of plane curves, there exists a point $m \in \ch(\mathcal{C})$ that is a depth median of $\mathcal{C}$, meaning
\begin{equation*}
    D(m, \mathcal{C}) = \sup_{q\in\ch(\mathcal{C})}D(q,\mathcal{C}) = \sup_{Q \in \Gamma} D(Q, \mathcal{C}).
\end{equation*}
\end{observation}

These observations, in combination with the preliminary discussion in Section~\ref{sec:decreasingDepth} on bounding a curve's depth by sample points along its length, allows for the formation of depth contours by sampling a region using points, e.g., the $\ch(\mathcal{C})$ region, which serve as depth estimators for curves within them.


\subsubsection{Decreasing Depth Relative to a Median Curve}\label{sec:decreasingDepth}
In the multivariate setting, it was outlined in \cite{Zuo:2000} that the depth of any point should ideally decrease monotonically with respect to a outwards translation from the deepest point. In place, the authors of \cite{Gijbels:2017} propose a generalization of this property for the functional data setting, which is reformulated and analyzed here.

For any sample of plane curves $\mathcal{C}$ and corresponding plane curve(s) of maximal depth $\med(\mathcal{C})$, a nondegenerate depth measure $D$ has \emph{decreasing depth relative to a deepest curve} provided
\begin{equation*}
    D(Q,\mathcal{C})\leq D(\med(\mathcal{C})+\alpha(Q-\med(\mathcal{C})),\mathcal{C})
\end{equation*}
for all plane curves $Q$ and $\alpha\in [0,1]$.

Note the quantity $D(\med(\mathcal{C}),\mathcal{C})$ represents a slight abuse of notation made for the sake of compactness. The quantity should properly be written in terms of any of the equally deep curves from the set $\med(\mathcal{C}) = \{C\in\mathcal{C} \mid D(C,\mathcal{C}) \geq D(S,\mathcal{C}) \text{ for all }S\in \Gamma\}$ from which we select and element.

Unfortunately, due to arbitrary sparsity of a sample of plane curves, $\mathcal{C}$, the curve stabbing depth cannot be guaranteed to satisfy this property. Figure~\ref{fig:decreasingDepthCounter} illustrates a simple counterexample for depth median points (degenerate curves) with respect to either query points (as depicted) or nondegenerate curves.

\begin{figure}[!ht]
    \centering
    \includegraphics[scale=0.8]{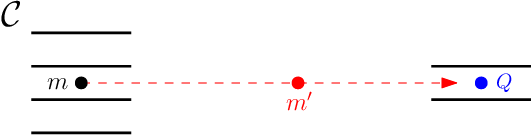}
    \caption{\small Composing $\mathcal{C}$ are six horizontal line segments drawn in black, which may be perturbed to satisfy assumptions of general position. The point $m$ is a median point of the given set $\mathcal{C}$, $Q$ is a degenerate query curve, and $m'$ a point along a affine path, drawn in red, from $m$ to $Q$. As depicted, it can be seen that $D(Q,\mathcal{C})\geq D(m',\mathcal{C}) = D(m+\alpha(Q-m),\mathcal{C})$ for some $\alpha$ corresponding to the position of $m'$ drawn, which is a violation of the desired inequality. In fact, $D(m',\mathcal{C})$ can be made arbitrarily close to zero by increasing the separation between the two groups of segments in $\mathcal{C}$.}
    \label{fig:decreasingDepthCounter}
\end{figure}

\subsection{Upper Semi-continuity}
\label{sec:semi-continuity}
Per a modification to the definition given in \cite{Reyes:2016}, with an additional requirement placed on the length of a perturbed curve, we define a depth measure $d$ operating on a set of plane curves $\mathcal{C}$ to be \emph{upper semi-continuous} if for any plane curve $Q$ and any $\epsilon>0$ there exists a $\delta>0$ such that 
\[
    \sup_{Q'\in \mathcal{Q}_\delta}d(Q',\mathcal{C})\leq d(Q,\mathcal{C})+\epsilon,
\]
where $\mathcal{Q}_\delta = \{Q'\in \Gamma \mid \dist(Q',Q)\leq\delta \text{ and } |L(Q)-L(Q')|\leq \delta\}$.

\begin{conjecture}\label{conj:semi-continuity}
Under the assumption of general position, where $Q$ has unique subpaths with respect to $\mathcal{C}$, the curve stabbing depth is upper semi-continuous; meaning, curve that are close to each other should score similar depth values.
\end{conjecture}

\subsection{Stability}
\label{sec:properties.stability}
Stability seeks to evaluate the degree, $k$, to which small changes in the input can lead to large changes in the output, where large values of $k$ are desirable, as they indicate insensitivity of the measure to perturbations in sample data. In our context, the input corresponds to the curves $Q$ and $\mathcal{C}$, and the output corresponds to the value of the depth measure $D(Q, \mathcal{C})$. Stability can be though of as a indicator of \textit{local} robustness; see Section~\ref{sec:robustness} for a more detailed discussion of robustness. 

Following \cite{Durocher:2009, Ramsay:2021}, we define the \emph{$k$-stability} of depth measure as follows.
A depth measure, $d$, is {\em $k$-stable} for a fixed $k>0$ if
\begin{equation*}
\forall \mathcal{C} \subseteq \Gamma \ \forall \mathcal{C}' \subseteq \Gamma \ \forall Q \in \Gamma \ \forall \epsilon > 0, \ \dist(\mathcal{C}, \mathcal{C}') \leq \epsilon \Rightarrow k\cdot|d(Q,\mathcal{C}) - d(Q,\mathcal{C}')| \leq \epsilon,
\end{equation*}
where the distance between two sets of curves $\mathcal{C}$ and $\mathcal{C}'$ is bounded by
\[ \dist(\mathcal{C}, \mathcal{C}') \leq \epsilon \Leftrightarrow \exists \text{continuous bijection } f:\mathcal{C} \to \mathcal{C}'
\text{ s.t. } \forall C \in \mathcal{C},\ \dist(C, f(C)) \leq \epsilon . \]



\begin{conjecture}\label{conj:stability}
There exists a fixed $k\in O(1)$ such that for any $\mathcal{C}$ and any $Q$ in $\Gamma$ with unique subpaths, the curve stabbing depth $D(Q,\mathcal{C})$ is $k$-stable. 
\end{conjecture}

An immediate consequence of Conjecture~\ref{conj:stability}, assuming it to be true, would be that a depth median point with respect to curve stabbing depth is $k'$-stable for some $k'\geq k$.

\subsection{Robustness}
\label{sec:robustness}
The robustness of a statistical measure indicates the degree to which estimators based on that measure are sensitive to the presence of outliers or perturbations in the data. One of the most common methods for assessing robustness, particularly when dealing with finite samples, is the \emph{breakdown point}; in this case, we focus on test statistics and the \emph{finite replacement breakdown point}, which is equal to the minimum number of data points in the worst case configuration that must be corrupted to cause the measure to take on an arbitrarily different value. 

The following formulation of the breakdown point is common, similar variations can be seen in \cite{zuo:2003} and \cite{Ramsay:2021}. See \cite{Huber:2009} for a treatise on robustness and breakdown points. 

Given any finite sample $X$ of size $n$ and an appropriate test statistic $T$, for which we later substitute our depth measure, the \emph{breakdown point} of $T$ for $X$, denoted $\epsilon^*(T,X)$, is the minimum number data points in $X$ that need to be replaced to form a corrupted set $X^{(m)}$ such that $T(X^{(m)})$ can be made arbitrarily different from $T(X)$. Specifically, the replacement finite sample breakdown point is defined as
\begin{equation*}
    \epsilon^*(T,X) = \frac{1}{n}\argmin_{1\leq m\leq n}\left\{\sup_{X^{(m)}}\dist(T(X),T(X^{(m)}))=\infty\right\},
\end{equation*}
where $X^{(m)}$ differs from $X$ in $m$ entries.

The geometric interpretation of robustness in the multivariate point case is quite straightforward, being the number of points that must be moved off to infinity before the point of maximal depth ceases to be such. However, in the functional or more general curve setting, considerations much be made to the shape of the curves involved in addition to location changes. Thus, a refinement of this interpretation is not immediately forthcoming. 

With all this said, under the given definition of robustness the following observation can be made with respect to depth median points.

\begin{observation}
For any $n>0$, there exists a set $\mathcal{C}\in \Gamma$ of cardinally $n$ for which the depth median $\med(\mathcal{C})$, defined with respect to $D(q,\mathcal{C})$, has breakdown point $\epsilon^*(\med(\mathcal{C}),\mathcal{C}) = 1/n$. 
\end{observation}
\begin{proof}
A worst-case example that realizes this breakdown point for any $n$ is illustrated in Figure~\ref{fig:robustness}. Generalization of the same construction used in the figure to arbitrary $n$ is straightforward.
\begin{figure}[!ht]
    \centering
    \includegraphics[scale=1]{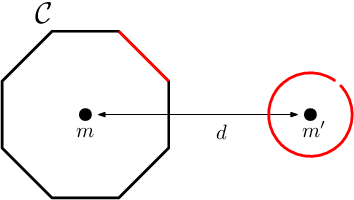}
    \caption{\small The set of curves $\mathcal{C}$ is composed of individual line segments along a regular eight sided polygon; more generally we may consider any regular $n$ sided polygon. All of the deepest points with respect to the population $\mathcal{C}$ are contained within the convex hull of $\mathcal{C}$. One of these is drawn and labeled $m$. The set $\mathcal{C}$ is then corrupted by replacement of the red line segment with a looped curve, drawn to the right, at a distance $d$ away from $m$. This results in the set of curves $\mathcal{C}^{(1)}$ with new median point $m'$. The distance $d$ between $m$ and $m'$ can be made arbitrarily large. Thus, this represents a worst-case breakdown point of $1/n$.}
    \label{fig:robustness}
\end{figure}
\end{proof}

Under the restriction of functional data, it is strongly suspected by the authors that the robustness of the depth median point is nearer to one half of the maximal depth (in terms of non-normalized depth).

\section{Monte Carlo Approximation}
\label{sec:approx}

Definition~\ref{def:curveDepth} initially suggests that efficient approximate computation by Monte Carlo methods is likely possible using a random sample of rays rooted along the query curve $Q$. In this section we explore three techniques to develop randomized algorithms that approximate curve stabbing depth. Firstly, we specify the exact form of the curves we consider, along with methods for approximate representation, e.g., discretization of continuous curves via polyline approximations or curve fitting to time sequence points. Secondly, we develop statistical methods for the application of random sampling and bounding the expected quality of approximation that can be achieved. Lastly, given efficient methods for computing ray-curve intersections for determining stabbing numbers, we analyze a worst-case bound on the performance of the constructed random sampling approximation method, and we compare our approximation algorithm against the exact algorithm proposed in Section~\ref{sec:alg} for the case of polylines.

\subsection{Data Representation and Approximation}
Section~\ref{sec:preliminaries.unparameterizedCurveSpace} discusses the sample space from which curves are drawn, along with polylines, which are effectively linear spline interpolants. We limit the ensuing discussion to curves drawn from $\Gamma$, and corresponding polyline approximations as similarly utilized in the Monte Carlo method proposed in \cite{Micheaux:2021}, with the obvious extension to fitting polylines to sequential (possibly time-stamped) data. Restricting input curves to polylines also allows for direct comparison between our exact algorithm and developed Monte Carlo techniques for computing curve stabbing depth.

\subsection{Monte Carlo Method}

This section develops the statistical backbone of the proposed Monte Carlo method. The notion of an \emph{oracle} is used in several places, which serves three purposes. The oracle is capable of each of the following: (1) returning (pseudo-)random points that are uniformly distributed along a curve according to its underlying representative parameterization, (2) returning a ray rooted at a specified point with a random angle taken from the $U[0,\pi)$ distribution, and (3) shooting rays into $\mathcal{C}$ and returning the value of the stabbing number for each such ray, all in $O(1)$ time. This oracle allows decoupling the existence of numerical methods and their performance from the initial discussion of the sampling scheme; additionally, it is not the goal of this paper to reiterate the many commonly used curve sampling methods that exists. Afterwards, known results are substituted in place of the oracle to yield more concrete bounds on the method's performance.

To begin, let $\mathcal{C}=\{C_1,C_2,\dots,C_n\}$ denote a set of curves from $\Gamma$ and let $Q$ be a curve from $\Gamma$. We are interested in computing $D(Q,\mathcal{C})$. In the Monte Carlo approximation scheme outlined below, we start by considering a sampling method built using uniform samples (produced via the oracle) of randomly oriented pairs of rays $\bm{\overrightarrow{q_\theta}}$ and $\bm{\overrightarrow{q_{\theta+\pi}}}$ rooted at points $\bm{q}$ selected uniformly at random along $Q$ 
with a uniformly random elevation $\bm{\theta}$ selected from $[0,\pi)$. These rays are then queried for their intersections with elements of $\mathcal{C}$ to approximate $D(Q,\mathcal{C})$. 
Later on, an approximate sampling method is developed using polyline approximations of smooth curves to handle cases where sampling cannot be easily performed (normally, because the uniform generation of points on $Q$ is either infeasible or impossible if an analytic parameterization is not known). To reduce possible ambiguity, note that the latter approximation step is unnecessary when the curves in question are  polylines. 

Recall the curve stabbing depth of a curve $Q$ taken with respect to a set $\mathcal{C}$ as defined in Eq.~\eqref{eq:depth}. Approaching this calculation from an alternative view, observe that it can be reformulated in terms of the expected stabbing number for a randomly selected ray along $Q$, the random selection being done uniformly. Adopting this view, the depth of $Q$ can be written as a linear combination of probabilities of the form 
\begin{equation*}
    p^{(k)} = \frac{1}{\pi L(Q)}\int_{q\in Q}\int_{0}^{\pi} \mathbb{I}\{\min\{\stab_{\mathcal{C}}(\overrightarrow{q_\theta}),\stab_{\mathcal{C}}(\overrightarrow{q_{\theta+\pi}})\} = k\} \,d\theta\, dq,
\end{equation*}
for $k=0,1,\dots,n$. We note that each of the above terms corresponds to the probability that a uniformly randomly generated ray $\overrightarrow{\bm{q_{\theta}}}$ rooted along $Q$, where $q$ is uniformly selected along $Q$ and $\theta$ is uniformly selected on $[0,\pi)$, has stabbing number $0,1,\dots,n$, respectively. 

Now, noting that
$$
\min\{\stab_{\mathcal{C}}(\overrightarrow{q_\theta}),\stab_{\mathcal{C}}(\overrightarrow{q_{\theta+\pi}})\}
= \sum_{k=0}^n k \, \mathbb{I}\{\min\{\stab_{\mathcal{C}}(\overrightarrow{q_\theta}),\stab_{\mathcal{C}}(\overrightarrow{q_{\theta+\pi}})\} = k\},
$$
we see the depth of a curve can be evaluated as
\begin{equation*}
    D(Q,\mathcal{C}) = \sum_{k=0}^n k \, p^{(k)}.
\end{equation*}
By learning estimates $\bm{\hat{p}^{(0)}},\bm{\hat{p}^{(1)}},\dots,\bm{\hat{p}^{(n)}}$ of the true values of $p^{(0)},p^{(1)},\dots,p^{(n)}$ via random uniform ray sampling along $Q$, the depth of $Q$ can be approximated as
\begin{equation*}
    \hat{D}(Q,\mathcal{C}) = \sum_{k=0}^n k \, \bm{\hat{p}}^{(k)}.
    \label{eq.D.hat}
\end{equation*}
This is the basis of the Monte Carlo approximation scheme we now propose.
In particular, for given $Q$ and $\mathcal{C}$, consider generating a random sample
$$
(\bm{\overrightarrow{q_\theta}})_1,(\bm{\overrightarrow{q_\theta}})_2,\ldots,(\bm{\overrightarrow{q_\theta}})_N
$$
of rays drawn uniformly and independently along $Q$; use the oracle to first generate a sample of random points $\bm{q}_0,\bm{q}_1,\dots,\bm{q}_N$ along $Q$ by utilizing the underlying representative parameterization and drawing from the $U[0,1]$ distribution to identify a point $\bm{q}$ on $Q$, and then respectively associate to each a random angle $\bm{\theta}_0,\dots,\bm{\theta}_N$ drawn from the $U[0,\pi)$ distribution.
The rays $\bm{\overrightarrow{q_\theta}}$ and $\bm{\overrightarrow{q_{\theta+\pi}}}$ are then shot into $\mathcal{C}$, whereby the oracle returns the stabbing numbers of each and we may compute the minimum of the two rays' stabbing numbers. From this, the resulting minimum stabbing numbers counts for the sample of rays can be obtained as
\begin{equation*}
    \bm{s^{(k)}} = \sum_{i=1}^N\mathbb{I}\{\min\{\stab_{\mathcal{C}}((\bm{\overrightarrow{q_\theta}})_i),\stab_{\mathcal{C}}((\bm{\overrightarrow{q_{\theta+\pi}}})_i)\} = k\}
\end{equation*}
for $k=0,1,\dots,n$. The estimates required for the construction of $\hat{D}(Q,\mathcal{C})$ given in Eq.~\eqref{eq.D.hat} are then simply
\begin{equation*}
    \bm{\hat{p}^{(k)}} = \frac{\bm{s^{(k)}}}{N},
\end{equation*}
the proportion of the total number of sampled rays that obtain the corresponding stabbing numbers.

Borrowing from Lemma~7 in \cite{Durocher:2017}, and exploiting a similar result from \cite{Devroye:1985}, we derive a bound on the probability of closeness of the resulting approximation based on the number of random rays used.

\begin{theorem}\label{thm:propClose}
For any finite set of continuous curves $\mathcal{C}\subseteq \Gamma$ and a query curve $Q\in \Gamma$ and any $\epsilon>0$, if the number of sampled rays $N$ is sufficiently large, namely $N\geq 20(n+1)n^2/\epsilon^2$, then 
\begin{equation*}
    \mathbb{P}\left(|\hat{D}(Q,\mathcal{C})-D(Q,\mathcal{C})|<\epsilon\right)\geq 1- 3e^{-\frac{N\epsilon^2}{25n^2}}.
\end{equation*}
\end{theorem}
\begin{proof}
Begin by observing that
\begin{align*}
     \mathbb{P}\left(|\hat{D}(Q,\mathcal{C})-D(Q,\mathcal{C})|<\epsilon\right) 
     &= \mathbb{P}\left(\left|\sum_{k=0}^n k\bm{\hat{p}^{(k)}} - \sum_{k=0}^n kp^{(k)}\right|<\epsilon\right)\\ 
     &\geq \mathbb{P}\left(\sum_{k=0}^n k\left|\bm{\hat{p}^{(k)}} - p^{(k)}\right|<\epsilon\right)\\ 
     &\geq \mathbb{P}\left(n\sum_{k=0}^n \left|\bm{\hat{p}^{(k)}}-p^{(k)}\right|<\epsilon\right)\\ 
     &= \mathbb{P}\left(\sum_{k=0}^n \left|\bm{s^{(k)}}-\mathbb{E}[\bm{s^{(k)}}]\right|<\frac{N}{n}\epsilon\right).
\end{align*}
The result then follows by Lemma~3.1 from \cite{Devroye:1985}, since  $\bm{s^{(0)}},\bm{s^{(1)}},\dots,\bm{s^{(n)}}$ are distributed according to a multinomial distribution with parameters $N$ and $(p^{(0)},p^{(1)},\dots,p^{(n)})$, as required.
\end{proof}

Reiterating the results of Theorem~\ref{thm:propClose}, we note that for a desired level of precision $\epsilon>0$, we require $N\in \Omega(n^3/\epsilon^2)$ samples to have a probability of closeness greater than $1- 3e^{-\frac{N\epsilon^2}{25n^2}}$. While this result does not preclude tighter bounds requiring fewer samples, Theorem~\ref{thm:propClose} provides a baseline indicator on performance.

\begin{corollary}
For any finite set of continuous curves $\mathcal{C}\subseteq \Gamma$, a query curve $Q\in \Gamma$, and any $\epsilon>0$, the depth estimate $\hat{D}(Q,\mathcal{C})$ converges in probability to $D(Q,\mathcal{C})$; that is, 
\begin{equation*}
    \lim_{N\to \infty}\mathbb{P}\left(|\hat{D}(Q,\mathcal{C})-D(Q,\mathcal{C})|<\epsilon\right) = 1.
\end{equation*}
\end{corollary}
\begin{proof}
As $N\to \infty$ independently of $n$ in the current context, the bound given in Theorem~\ref{thm:propClose} is directly applicable for sufficiently large $N$. Convergence in probability then follows directly.
\end{proof}

Until now, our discussion in Section~\ref{sec:approx} has considered exact ray sampling along $Q$ in $\mathcal{C}$; next we examine  approximating the depth of $Q$ using a sample of rays shot into a set of polyline approximations of curves in $\mathcal{C}$, which is more in keeping with the input domain of our exact algorithm.
Namely, for each $C_i\in\mathcal{C}$ we construct a polyline approximation $P_i=(p_{i0},p_{i1},p_{i2},\dots,p_{im})$ anchored at sample points along the length of $C_i$ provided by the oracle, and likewise let $P_Q=(q_0,\dots,q_m)$ be a similarly constructed polyline approximation of $Q$. Assume that each polyline approximation has $m$ segments; the ensuing statements easily generalize to polylines composed of $O(m)$ segments. The sampling scheme utilized in this scenario follows a slight modification of that outlined above. Namely, rays are randomly sampled along $P_Q$ by selecting a segment $\bm{I}$ of the form $I_{i}=\overline{p_{i}p_{i+1}}$ with probability
\begin{equation*}
    \mathbb{P}(\bm{I}=I_{i}) = \frac{L(I_{i})}{L(P_Q)} = \frac{L(I_{i})}{\sum_{j=0}^{m-1}L(I_{j})},
\end{equation*}
which is the proportion of the total length of $P_Q$ contributed by $I_i$. Secondly, a point $\bm{q}$ along the selected segment $\bm{I}$ is generated randomly and uniformly through affine interpolation along the segment by utilizing the $U[0,1]$ distribution, with a random angle $\bm{\theta}$ selected from $U[0,\pi)$ as before. As before, the resulting rays $\bm{\overrightarrow{q_\theta}}$ and $\bm{\overrightarrow{q_{\theta+\pi}}}$ are shot into $\mathcal{P}$ and the minimum of the two rays' stabbing numbers is found. Importantly, notice the estimates $\bm{\hat{p}^{(0)}},\bm{\hat{p}^{(1)}},\dots,\bm{\hat{p}^{(n)}}$ formed by a random sample of $N$ rays, in the way previously outlined, are no longer direct estimates of the true stabbing number probabilities along $Q$ with respect to $\mathcal{C}$, but rather of the probabilities for stabbing rays along $P_Q$ with respect to $\mathcal{P}$. 

Consequently, to bound the quality of approximation one can inspect the difference 
\begin{align}
    |\hat{D}(P_Q,\mathcal{P}) - D(Q,\mathcal{C})| 
    & \leq |\hat{D}(P_Q,\mathcal{P}) - D(P_Q,\mathcal{P})| + |D(P_Q,\mathcal{P})-D(Q,C)|\\
    & \leq |\hat{D}(P_Q,\mathcal{P}) - D(P_Q,\mathcal{P})| \label{eq:D.L}\\ 
    & \quad + |D(P_Q,\mathcal{P})-D(Q,\mathcal{P})|\label{eq:D.C}\\
    & \quad + |D(Q,\mathcal{P})-D(Q,\mathcal{C})|, \label{eq:D.S}
\end{align}
where $\hat{D}(P_Q,\mathcal{P})$ denotes the $N$ sample estimate of $D(Q,\mathcal{C})$ utilizing $m$ segment polyline approximations of both $Q$ and $\mathcal{C}$. In the expansion, Eq.~\eqref{eq:D.L} embodies the variance (or Monte Carlo simulation error) with the sum between Eq.~\eqref{eq:D.C} and Eq.~\eqref{eq:D.S} bounding the bias of the overall approximation. 

The value of Eq.~\eqref{eq:D.L} can be bounded in probability via the result obtained in Theorem~\ref{thm:propClose}, while Eq.~\eqref{eq:D.C} relates to the continuity of $Q$ with respect to the sample of polyline approximates, as discussed in Section~\ref{sec:semi-continuity}, and Eq.~\eqref{eq:D.S} can be bounded according to the stability of curve stabbing depth from Section~\ref{sec:properties.stability}. As all but the first do not have known tight bounds, we defer derivation of a strict bound on this quantity to future research.

\subsubsection{Computation of Ray Intersection Queries}
As the stabbing number of each sampled ray in the Monte Carlo approximation must be calculated, it follows that the required query time for counting such ray intersections with a set of smooth curves (and later polylines) must be weighed against the desired precision in order to determine the number of sample points available before performance becomes untenable. For the case of polylines, the threshold at which running our exact algorithm becomes more efficient than the approximation scheme with respect to a desired precision, and requisite number of samples, should be found.

The discussion surrounding methods for shooting rays into arbitrary smooth curves drawn from $\Gamma$ leads outside the scope of this paper and into the realm of numerical analysis, so is not elaborated on further as we are primarily concerned with points of comparison against our exact algorithm for computing curve stabbing depth, which itself is not directly applicable to smooth curves without approximation. Such methods exist for several classes of curves, e.g., using a quadrature approximation of the curve used to determine intersection points, similar to those discussed below. 

For randomly positioned and oriented ray-polyline intersections, it is impractical to  use the dynamic tangent point and stabbing number data structures used in the exact algorithm, as these would add considerable overhead requiring all the update points between query rays to be computed and traversed in a continuous fashion. Instead, we use a different method for detecting ray intersections, better suited to performing multiple independent queries. This particular problem of counting the number of intersections between a ray and a set of distinct geometric objects composed of line segments is an instance of the \emph{colored segment intersection searching} problem. The best solutions known to the authors for solving this problem are those proposed by \cite{Agarwal:1996} and \cite{Gupta:1994}, the latter of the two methods being capable of performing a single query in $O(\log^2 (nm) + k\log(nm))$ time, where $k\in O(nm)$ is the number of distinct polyline segments intersected, plus an initial $O((nm+\chi)\log(nm))$ preprocessing worst-case time and space to construct the query data structure, where $\chi \in O(n^2 m^2)$ denotes the number of pairwise intersections. In the worst case, a query ray always intersects all $n$ distinct curves in the sample, corresponding to $k=nm$ in the forgoing bound, informing us through comparison with the time required for computing the exact algorithm, that we can afford at most $N\in O\left(\frac{n^3 + n^2m\log^2 m +nm^2\log^2 m}{\log^2(nm)+nm\log(nm)}\right)$ ray queries in the worst case before our Monte Carlo approximation becomes asymptotically less efficient than exact calculation, not accounting for any further processing required. 

To build intuition of this threshold on the number of ray query samples, consider the case $m\in O(n)$, which simplifies the above expression to $O(n^3)$. Therefore, when $m \in O(n)$, the Monte Carlo method is only faster than the exact algorithm when the number of samples is $o(n^3)$, which is only slightly lower than the $O(n^3/\epsilon)$ samples required to attain a guarantee on the probability of closeness given by Theorem~\ref{thm:propClose}. This, suggests that Monte Carlo, in the form described, may not be an efficient method of approximating the Curve Stabbing Depth.

\section{Discussion and Directions for Future Research}
\label{sec:conclusion}

In this section we briefly discuss possible generalizations of curve stabbing depth to higher dimensions, other approaches to calculating the depth of curves for both continuous and discretized curves, and additional possible directions for future research.

\subsection{Stability and Upper Semi-Continuity}
In Conjecture~\ref{conj:semi-continuity}, the authors hypothesize that curve stabbing depth is upper semi-continuous, and in Conjecture~\ref{conj:stability}, that it is also stable, i.e., that it has stability bounded by a constant. Proving (or identifying counter-examples to) Conjectures~\ref{conj:semi-continuity} and~\ref{conj:stability} remains open.

\subsection{Batched Depth Calculations}
As there is considerable overlap in the prepossessing stages for computing the depth of curves relative to the same sample, it might be feasible to accelerate \emph{batched queries}, as examined for select depth measures in \cite{Durocher:2022-2}.

\subsection{Applications to Classification and Clustering}
Based on empirical analysis of various examples, curve stabbing depth appears to implicitly weight points nested within clusters more than those that would otherwise be considered deep, according to other multivariate depth measures, within the larger cloud of curves. As such, investigation into the effectiveness of clustering curves using curve stabbing depth could prove fruitful in domains where some form of combined classification and depth information is desired, e.g., for outlier detection or anomalous data detection.

\subsection{Generalizations to Higher Dimensions}
\label{sec:conclusion.higherDim}

When a curve $Q$ and a set $\cal C$ of curves lie in a $k$-dimensional flat of $\mathbb{R}^d$ for some $k < d$, the $d$-dimensional curve stabbing depth of $Q$ as calculated using a ray relative to $\cal C$ is zero, whereas the $k$-dimensional curve stabbing depth of $Q$ relative to $\cal C$ is non-zero in general, meaning that the straightforward application of Definition~\ref{def:curveDepth} as expressed is not consistent across dimensions. 

Alternatively, a natural generalization of Definition~\ref{def:curveDepth} to higher dimensions is to replace the rotating stabbing ray by a rotating $k$-dimensional half-hyperplane, and to measure the number of curves it intersects as it rotates. Such generalizations could be explored to define measures of data depth for curves in higher dimension that remain consistent across dimensions, and reduce to the expression given in the planar case in Definition~\ref{def:curveDepth} for co-planar instances.

\subsection{Other Possible Measures of Data Depth for Sets of Curves}
\label{sec:conclusion.otherMeasures}
Alt and Godau \cite{Alt:1995} gave an algorithm to compute the Fr\'{e}chet distance between two $m$-vertex polylines in $O(m^2 \log m)$ time . Applying their algorithm iteratively, we can compute the mean Fr\'{e}chet distance between an $m$-vertex polyline $Q$ and a set $\mathcal{P}$ of $n$ polylines, each with $O(m)$ vertices, in $O(nm^2 \log m)$ time; this time is of similar magnitude to that of our algorithm for computing the curve stabbing depth of $Q$ relative to $\mathcal{P}$: $O(n^2m\log^2 m + nm^2 \log^2 m)$ (Theorem~\ref{thm:mainAlg}). 
The mean Fr\'{e}chet distance can be used to define a depth measure for sets of curves (and polylines), as
\begin{equation}
\label{eq:otherMeasures}
F(Q, \mathcal{C}) = \frac{1}{|\mathcal{C}|} \sum_{C \in \mathcal{C}} \dist(Q, C) .    
\end{equation}
A curve that minimizes $F(\cdot, \cdot)$ in Eq.~\eqref{eq:otherMeasures} generalizes the notion of the Weber point to sets of curves, as well as Type B depth functions, as defined by Zuo and Serfling \cite{Zuo:2000} for sets of points.  
Just as a median of a set $\mathcal{S}_1$ of points in $\mathbb{R}$ minimizes the average distance to points in $\mathcal{S}_1$,
recall that the Weber point, $w_{\mathcal{S}_2}$, of a set $\mathcal{S}_2$ of points in $\mathbb{R}^2$ minimizes the average Euclidean distance to points in $\mathcal{S}_2$. Unlike a one-dimensional median which can be computed 
using a linear-time selection algorithm,
the exact position of $w_{\mathcal{S}_2}$ cannot be computed exactly in general \cite{Bajaj:1988}, and must be approximated \cite{Bose:2003,Durocher:2009}.
The Weber point is highly unstable \cite{Durocher:2009,Durocher:2006}, i.e., it is not stable for any fixed $k$ (see Section~\ref{sec:properties.stability}), implying the same property for $F$.
To define a depth measure whose value is high when a curve is deep relative to $\mathcal{C}$ and is low for outliers relative to $\mathcal{C}$ requires taking an inverse or a difference, as is done for Type B depth functions in \cite{Zuo:2000}, e.g.,
\begin{equation*}
F'(Q, \mathcal{C}) = \frac{1}{1+F(Q,\mathcal{C})} 
= \frac{|\mathcal{C}|}{|\mathcal{C}| + \sum_{C \in \mathcal{C}} \dist(Q, C)}.
\end{equation*}
Such a depth measure, however, does not have a bounded region of non-zero depth (see Section~\ref{sec:properties.non-zeroDepth}).
$F'$ can be made invariant under similarity transformations by adding a scaling factor, e.g., scale by the inverse of the radius of the smallest enclosing ball of $\mathcal{C}$ (see Section~\ref{sec:properties.affine}). 

\section*{Acknowledgements}
The authors thank Kelly Ramsay for suggesting a helpful reference related to the discussion in Section~\ref{sec:conclusion.otherMeasures}.


\bibliographystyle{abbrv}
\bibliography{bibliography.bib}

\end{document}